\input ./style/arxiv-general.cfg
\documentclass[MSNbibl,number,citesort,dvips]{arxbj}
\makeatletter
   \@ifpackageloaded{graphicx}{}{\usepackage{graphicx}}
\makeatother
\usepackage{algorithm,algorithmic}

\aid{0}
\volume{21}
\issue{3}
\pubyear{2015}
\firstpage{1855}
\lastpage{1883}
\doi{10.3150/14-BEJ629} 

\makeatletter
\newcommand{\rrvert}{\vert}
\newcommand{\llvert}{\vert}
\def\prettyref{\eqref}
\newcommand{\eqref}[1]{(\ref{#1})}
\newtheorem{thmm}{Theorem}
\newtheorem{prop}[thmm]{Proposition}
\newtheorem{lem}[thmm]{Lemma}
\newremark{example}[thmm]{Example}

\newremark{remark}[theorem]{Remark}

\newproclaim{assumption}{Assumption}
\renewcommand{\P}{\mathbb{P}} 
\newcommand{\E}{\mathbb{E}} 
\newcommand{\setX}{\mathcal{X}}
\newcommand{\setY}{\mathcal{Y}}
\newcommand{\C}{\mathcal{C}}
\newcommand{\R}{\varrho} 
\newcommand{\ev}{\mathcal{A}} 

\makeatother

\begin{document}
\begin{frontmatter}

\title{On particle Gibbs sampling}
\runtitle{On particle Gibbs sampling}

\begin{aug}
\author[a]{\inits{N.}\fnms{Nicolas}~\snm{Chopin}\corref{}\thanksref{a}\ead[label=e1]{nicolas.chopin@ensae.fr}}
\and
\author[b]{\inits{S.S.}\fnms{Sumeetpal S.}~\snm{Singh}\thanksref{b}\ead[label=e2]{sss40@eng.cam.ac.uk}}
\address[a]{CREST-ENSAE and HEC Paris, 3 Avenue Pierre Larousse, 92235 Malakoff,
France.\\ \printead{e1}}
\address[b]{Department of Engineering, University of Cambridge,
Trumpington Street,
Cambridge, CB2 1PZ,
UK. \printead{e2}}
\end{aug}

\received{\smonth{1} \syear{2014}}


\begin{abstract}
The particle Gibbs sampler is a Markov chain Monte Carlo (MCMC) algorithm
to sample from the full posterior distribution of a state-space model.
It does so by executing Gibbs sampling steps on an extended target
distribution defined on the space of the auxiliary variables generated
by an interacting particle system. This paper makes the following
contributions to the theoretical study of this algorithm.
Firstly,
we present a coupling construction between two particle Gibbs updates
from different starting points and we show that the coupling probability
may be made arbitrarily close to one by increasing the number of particles.
We obtain as a direct corollary that the particle Gibbs kernel is
uniformly ergodic.
Secondly, we show how the inclusion
of an additional Gibbs sampling step that reselects the ancestors
of the particle Gibbs' extended target distribution, which is a popular
approach in practice to improve mixing, does indeed yield a theoretically
more efficient algorithm as measured by the asymptotic variance.
Thirdly, we extend particle Gibbs to work with
lower variance resampling schemes.
A detailed numerical study is provided to demonstrate the efficiency
of particle Gibbs and the proposed variants.
\end{abstract}

%
\end{frontmatter}

\section{Introduction}

PMCMC (particle Markov chain Monte Carlo \cite{PMCMC}) is a new
set of MCMC algorithms devised for inference in state-space models
which has attracted considerable attention in statistics. It has in
a short time triggered intense scientific activity spanning methodological
\cite{Silva2009,WhiteleyChange2010,smc2,LindstenAncestor} and applied
work, the latter in domains as diverse as ecology \cite{Peters2010a},
electricity forecasting \cite{launay2012particle}, finance \cite{Pitt2012134},
systems biology \cite{golightly2011bayesian}, social networks \cite{Everitt2012}
and hydrology \cite{Vrugt2012}. One appeal of PMCMC is that it makes
it possible to perform ``plug-and-play'' inference for complex hidden
Markov models, that is, the only requirement is that one needs to
be able to sample from the Markov transition of the hidden chain,
which is in most cases non-demanding, in contrast to previous
approaches based on standard MCMC.

Each PMCMC step generates an interacting particle system; see \cite{DouFreiGor,delMoral:book} and \cite{CapMouRyd} for general references
on particle algorithms (also known as Sequential Monte Carlo algorithms).
Several instances of PMCMC may be analysed as \textit{exact Monte
Carlo approximations} of an ideal algorithm, that is, as a noisy version
of an ideal algorithm where some intractable quantity is replaced
by an unbiased Monte Carlo estimate (computed from the interacting
particle system). Such algorithms are analysed in detail in \cite{AndrieuVihola}.
The term `exact' in the phrase `exact Monte Carlo' highlights the
fact that, despite being an approximation of an ideal algorithm, PMCMC
samples exactly from the distribution of interest.

However, this interpretation does not seem applicable to variants of
PMCMC involving a particle Gibbs step. While particle Gibbs also generates
a complete interacting particle system at each iteration, it does
so conditionally on the trajectory for one particle being fixed, and
it does not replace an intractable quantity of an ideal algorithm
with an unbiased estimator.

The objective of this paper is to undertake a theoretical study of
particle Gibbs to try to support its very favourable performance observed
in practice. For this, we design a coupling construction between two particle
Gibbs updates that start from different trajectories and establish
that the coupling probability may be made arbitrarily large by increasing
the number of particles $N$. As a direct corollary, we conclude that
the transition kernel of particle Gibbs is uniformly ergodic (under
suitable conditions). This strong result supports why particle Gibbs
can be expected, and does indeed, perform so well in practice. Our
coupling construction is maximal for some special cases and appears
unique in the literature on particle systems.

Secondly, we show how the inclusion of an additional backward sampling
step that reselects the ancestors of the particle Gibbs' extended
target distribution, first proposed by \cite{Whiteley_disc_PMCMC}
and now a popular approach in practice to improve mixing \cite{Lindsten2012},
does indeed yield a theoretically more efficient algorithm as measured
by the asymptotic variance of the central limit theorem. Thirdly,
and as another way to enhance mixing, we extend the original particle
Gibbs sampler (which is based on the multinomial resampling scheme
as presented in the original paper of \cite{PMCMC}) to work with
lower variance residual or systematic resampling schemes.
This variety of implementation of particle Gibbs raises an obvious
question: which variant performs best in practice? We present numerical
comparisons in a particular example, which suggests that the backward
sampling strongly improves the mixing of particle Gibbs, and, when
it cannot be implemented, then residual and systematic resampling
leads to significantly better mixing than multinomial resampling.

The plan of the paper is the following. Section~\ref{sec:PG_framework}
sets up the notation and defines the particle Gibbs algorithm.
This section reviews the original particle Gibbs algorithm of \cite{PMCMC}
and presents a reinterpretation of particle Gibbs as a
Markov kernel to facilitate the analysis to follow in the later sections.
Some supporting technical results are also presented. Section~\ref{sec:Main-result}
proves that the particle Gibbs kernel is uniformly ergodic. To that effect,
a coupling construction is obtained such that the coupling probability
between two particle Gibbs updates may be made arbitrarily large for $N$
large enough.
Section~\ref{sec:Backward-sampling} discusses the backward sampling step
proposed by \cite{Whiteley_disc_PMCMC}, and establishes dominance of particle
Gibbs with this backward sampling step over the version without.
Section~\ref{sec:alt-resampling} discusses how to extend particle
Gibbs to alternative resampling schemes.
Section~\ref{sec:Numerical-experiments} presents a numerical
comparison of
the variants of particle Gibbs discussed in the previous sections.
Section~\ref{sec:Conclusion} concludes.

\section{Definition of the particle Gibbs sampler}
\label{sec:PG_framework}


\subsection{Notation}
\label{sec:Notations}

For $m\leq n$, we denote by $m\dvtx n$ the range of integers $\{m,\ldots,n\}$,
and we use extensively the semicolon short-hand for collections of
random variables, for example, $X_{0:T}=(X_{0},\ldots, X_{T})$,
$X_{t}^{1:N}=(X_{t}^{1},\ldots,X_{t}^{N})$,
and even in an nested form, $X_{0:T}^{1:N}=(X_{0}^{1:N},\ldots,X_{T}^{1:N})$;
more generally $X_{t}^{v}$, where $v$ is a vector in $\mathbb{N}^{+}$
will refer to the collection $(X_{t}^{n})_{n\in v}$. These short-hands
are also used for realisations of these random variables, which are
in lower case, for example, $x_{0:t}$ or $x_{t}^{1:N}$. The sub-vector
containing
the $t$ first components of some vector $Z_{T}$ is denoted by $
[Z_{T} ]_{t}$.

For a vector $r^{1:N}$ of probabilities, $r^{n}\in[0,1]$ and $\sum_{n=1}^{N}r^{n}=1$,
we denote by $\mathcal{M}(r^{1:N})$ the multinomial distribution
which produces outcome $n$ with probability $r^{n}$, $n\in1\dvtx N$.
For reals $x,y$, let $x\vee y=\operatorname{max}(x,y)$ and $x\wedge
y=\operatorname{min}(x,y)$.
The integer part of $x$ is $ \lfloor x \rfloor$, and the
positive part is $x^{+}=x\vee0$. The cardinal of a finite set
$\mathcal{C}$
is denoted as $\llvert \mathcal{C}\rrvert $.

For a complete separable metric space $\setX$, we denote by $\mathcal
{P}(\setX)$ the set of probability distributions on $\setX$. For a probability
measure $\mu\in\mathcal{P}(\setX$), a kernel $K\dvtx\setX\rightarrow
\mathcal
{P}(\setX)$
and a measurable function $f$ defined on $\setX$, we use the following
standard notation: $\mu(f)=\int_{\setX}\,\mathrm{d}\mu f$, $Kf$ is the application
$x\rightarrow\int_{\setX}K(x,\mathrm{d}x')f(x')$, and $\mu K$ is the probability
measure $(\mu K)(A)=\int_{\setX}\mu(\mathrm{d}x)K(x,A)$. The atomic measure
at $a\in\setX$ is denoted by $\delta_{a}(\mathrm{d}x)$. We denote by $\mu
\otimes K$
the measure $\mu(\mathrm{d}x)K(x,\mathrm{d}x')$ on the product space $\setX\times\setX$.
Finally, we shall often use the same symbol for distributions and
densities; for example, $m_0(\mathrm{d}x_0)=m_0(x_0) \,\mathrm{d}x_0$ means that the
distribution $m_0(\mathrm{d}x_0)$ admits the function $x_0\rightarrow m_0(x_0)$
as a probability density relative to some sigma-finite dominating
measure $\mathrm{d}x_0$.

\subsection{The target distribution}
\label{sec:Model}

Let $\setX$ be a complete separable metric space, and
$(X_{t})_{t\geq0}$ a discrete-time $\setX$-valued Markov
chain, with initial law $m_{0}(\mathrm{d}x_{0})=m_{0}(x_{0}) \,\mathrm{d}x_{0}$, and
transition law $m_{t}(x_{t-1},\mathrm{d}x_{t})=m_{t}(x_{t-1},x_{t}) \,\mathrm{d}x_{t}$,
where $\mathrm{d}x_0$, $\mathrm{d}x_t$ are appropriately chosen (possibly identical)
sigma-finite dominating measures.
Let $(G_{t})_{t\geq0}$ be a sequence of $\setX\rightarrow\mathbb
{R}^+$ potential
functions. In the context of hidden Markov models, typically
$G_{t}(x_{t})=g(x_{t},y_{t})$,
the density (with respect to some dominating measure $\mathrm{d}y$) of observation
$y_{t}$ of the $\setY$-valued random variable $Y_{t}$, conditional
on state $X_{t}=x_{t}$.

It is convenient to work directly with the path model,
that is, we define $Z_{t}=X_{0:t}$ (and $z_{t}=x_{0:t}$) taking values
in $\setX^{t+1}$, and slightly abusing notation, we extend the domain
of $G_{t}$ from $\mathcal{X}$ to $\mathcal{X}^{t+1}$ as follows:
$G_{t}(z_{t})=G_{t}(x_{t})$. The $Z_{t}$'s form a time inhomogeneous
Markov kernel, with initial law $q_{0}(\mathrm{d}z_{0})=m_{0}(\mathrm{d}x_{0})$, and
transition
\[
q_{t}\bigl(z_{t-1},\mathrm{d}z_{t}^{\prime}\bigr)=
\delta_{z_{t-1}}\bigl(\mathrm{d}x_{0:t-1}^{\prime}\bigr)
m_{t}\bigl(x_{t-1},x_{t}^{\prime}
\bigr)\,\mathrm{d}x_{t}^{\prime}
\]
that is, keep all of $z_{t-1}$ and append new state $x_{t}$, from Markov
transition $m_{t}(x_{t-1},\mathrm{d}x_{t})$. The associated (Feynman--Kac)
path measures are
%
\begin{equation}
\mathbb{Q}_t(\mathrm{d}z_{t})=\mathbb{Q}_t(\mathrm{d}x_{0:t})=
\frac{1}{\mathcal
{Z}_t}G_0(x_0)m_{0}(\mathrm{d}x_{0})
\prod_{s=1}^{t} \bigl\{
G_{s}(x_{s})m_{s}(x_{s-1},\mathrm{d}x_{s})
\bigr\}, \label
{eq:invariant_dist}
\end{equation}
where $\mathcal{Z}_t$ is defined as
\[
\mathcal{Z}_t=\int_{\setX^{t+1}}G_0(x_0)m_{0}(\mathrm{d}x_{0})
\prod_{s=1}^{t} \bigl\{
G_{s}(x_{s})m_{s}(x_{s-1},\mathrm{d}x_{s})
\bigr\}
\]
assuming from now on that $0<\mathcal{Z}_t<+\infty$.
The target distribution to be sampled from is $\mathbb{Q}_T(\mathrm{d}z_{T})$ for
some fixed $T$, which can also be interpreted as the full posterior of
a state-space model.

The fact that we work directly with the path $Z_{t}$, and path-valued
potential functions $G_{t}(z_{t})$, reveals that our results could
be extended easily to the situation where in the original formulation
for $X_{t}$, the potential function depended on past values, for example,
$G_{t}(x_{t-1},x_{t})$. In that way, one may consider, for instance, more
general algorithms
where particles are mutated according to a proposal kernel that may differ
from the Markov kernel of the considered model.
However, in the only part of the paper (Section~\ref{sec:Backward-sampling}) where we shall revert to the original
formulation based on $X_{t}$, we will stick to the standard case
where $G_{t}$ depends only on $x_{t}$ for the sake of clarity.

Andrieu \textit{et al.} \cite{PMCMC} introduced an MCMC algorithm that samples from
(\ref{eq:invariant_dist})
by defining an extended target distribution (which admits (\ref
{eq:invariant_dist})
as its marginal) and then constructing a Gibbs sampler for this extended
target. In the next section, we review this construction of theirs.

\subsection{The extended target and the particle Gibbs sampler}
\label{sec:Probabilistic-description-of}

The starting point in the definition of \cite{PMCMC}'s extended
target distribution that admits \eqref{eq:invariant_dist} as its
marginal is the joint distribution of all the random variables generated
in the course of the execution of an (interacting) particle algorithm
that targets the path measures given in \eqref{eq:invariant_dist}.
We refer the reader to \cite{delMoral:book,CapMouRyd} for a review
of particle algorithms that target Feynman--Kac path measures.

The particle representation $\mathbb{Q}_t^{N}(\mathrm{d}z_{t})$ is the
empirical measure
defined as, for $t\geq0$,
\[
\mathbb{Q}_t^{N}(\mathrm{d}z_{t})=\frac{1}{N}
\sum_{n=1}^{N}\delta_{Z_{t}^{n}}(\mathrm{d}z_{t}),
\]
where the particles $Z_{t}^{1:N}=(Z_{t}^{1},\ldots,Z_{t}^{N})$ are
defined recursively as follows. First, $Z_{0}^{1:N}$ is obtained
by sampling $N$ times independently from $m_{0}(x_{0})\,\mathrm{d}x_{0}$. To
progress from time $t$ to time $t+1$, $t\geq0$, the pair
$(A_{t}^{1:N},Z_{t+1}^{1:N})$
is generated jointly from
\[
\varrho_t\bigl(z_{t}^{1:N},\mathrm{d}a_{t}^{1:N}
\bigr)\prod_{n=1}^{N}q_{t+1}
\bigl(z_{t}^{a_{t}^{n}},\mathrm{d}z_{t+1}^{n}\bigr),
\]
conditionally\vspace*{1.5pt} on $Z_{t}^{1:N}=z_{t}^{1:N}$, where the $A_{t}^{n}$'s,
$n\in1\dvtx N$, jointly sampled from the resampling distribution $\varrho_t
(z_{t}^{1:N},\mathrm{d}A_{t}^{1:N})$
are the ancestor variables, that is, $A_{t}^{n}$ is the label of the
particle at time $t$ which generated particle $Z_{t+1}^{n}$ at time
$t+1$. (Since the $A_t^n$'s are integer-valued, the dominating measure
of kernel $\varrho_t(z_t^{1:N},\mathrm{d}a_t^{1:N})$ is simply the counting measure.)

The law of the collection of random variables $(Z_{0:T}^{1:N},A_{0:T-1}^{1:N})$
generated from time 0 to some final time $T\geq1$ is therefore
\[
\vartheta_{T}^{N} \bigl(\mathrm{d}z_{0:T}^{1:N},\mathrm{d}a_{0:T-1}^{1:N}
\bigr)=m_{0}^{\otimes N}\bigl(\mathrm{d}z_{0}^{1:N}
\bigr)\prod_{t=1}^{T} \Biggl\{ \R
_{t-1}\bigl(z_{t-1}^{1:N},\mathrm{d}a_{t-1}^{1:N}
\bigr)\prod_{n=1}^{N} \bigl[q_{t}
\bigl(z_{t-1}^{a_{t-1}^{n}},\mathrm{d}z_{t}^{n}\bigr)
\bigr] \Biggr\}.
\]

The simplest choice for $\varrho_t$ is what is usually referred to as
the multinomial
resampling scheme, namely the $A_{t}^{n}$'s are drawn independently
from the multinomial distribution $\mathcal{M}
(W_{t}^{1:N}(z_{t}^{1:N}) )$,
where the $W_{t}^{n}$'s are the normalised weights
%
\begin{equation}
W_{t}^{n}\bigl(z_{t}^{1:N}\bigr)
\stackrel{\Delta} {=}\frac
{G_{t}(z_{t}^{n})}{\sum_{m=1}^{N}G_{t}(z_{t}^{m})}, \qquad
n\in1\dvtx N.\label{eq:def_weights}
\end{equation}
Then
$\varrho_t(z_{t}^{1:N},\mathrm{d}a_{t}^{1:N})$,
$t\geq0$ equals
%
\begin{equation}
\varrho_t\bigl(z_{t}^{1:N},\mathrm{d}a_{t}^{1:N}
\bigr)= \Biggl\{\prod_{n=1}^{N}W_{t}^{a_{t}^{\mathrm{n}}}
\bigl(z_{t}^{1:N}\bigr) \Biggr\} \,\mathrm{d}a_t^{1:N}.
\label{eq:multiRS}
\end{equation}

For now, we assume this particular choice for $\varrho_t$, and our main
results will therefore be specific to multinomial resampling. Note,
however, that we will discuss alternative
resampling schemes at the end of the paper; see Section~\ref{sec:alt-resampling}.

We now state an intermediate result which is needed to ensure the
validity of the extended target \eqref{eq:defpi} below.

\begin{prop} One has
%
\begin{equation}
\E_{\vartheta_{T}^{N}} \Biggl[\prod_{t=0}^{T}
\Biggl\{ \frac{1}{N}\sum_{n=1}^{N}G_{t}
\bigl(Z_{t}^{n}\bigr) \Biggr\} \Biggr]=\mathcal{Z}_T.\label{eq:unbiased}
\end{equation}
\end{prop}

See, for example, Lemma~3 in \cite{DelMoral1996unbiased}.
In order to state the particle Gibbs sampler and prove it leaves
$\mathbb{Q}_T(\mathrm{d}z_{T})$
invariant, we commence first with the definition of the following
extended distribution $\pi_{T}^{N}$ of \cite{PMCMC} whose sampling
space is the sampling space of the measure $\vartheta_{T}^{N}$ augmented
to include a discrete random variable $N^{\star}\in1\dvtx N$,
%
\begin{eqnarray}
\label{eq:defpi} &&\pi_{T}^{N} \bigl(\mathrm{d}z_{0:T}^{1:N},\mathrm{d}a_{0:T-1}^{1:N},\mathrm{d}n^\star
\bigr)\nonumber\\
&&\quad = \frac{1}{\mathcal{Z}_T} \vartheta_{T}^{N}
\bigl(\mathrm{d}z_{0:T}^{1:N},\mathrm{d}a_{0:T-1}^{1:N} \bigr)
\Biggl[\prod_{t=0}^{T-1} \Biggl\{
\frac{1}{N}\sum_{n=1}^{N}G_{t}
\bigl(Z_{t}^{n}\bigr) \Biggr\} \Biggr] \frac{1}{N}
G_T\bigl(z_T^{n^\star}\bigr)
\nonumber
\\
&&\quad =\frac{1}{\mathcal{Z}_T}m_{0}^{\otimes N}\bigl(\mathrm{d}z_{0}^{1:N}
\bigr)
\\
&& \qquad{}\times\prod_{t=1}^{T}
\Biggl[ \Biggl\{ \frac{1}{N}\sum_{n=1}^{N}G_{t-1}
\bigl(z_{t-1}^{n}\bigr) \Biggr\} \prod
_{n=1}^{N} \bigl\{W_{t-1}^{a_{t-1}^n}
\bigl(z_{t-1}^{1:N}\bigr)\,\mathrm{d}a_{t-1}^n
q_{t}\bigl(z_{t-1}^{a_{t-1}^{n}},\mathrm{d}z_{t}^{n}
\bigr) \bigr\} \Biggr]\nonumber
\\
& &\qquad{}\times\frac{1}{N}G_{T}\bigl(z_{T}^{n^\star}
\bigr),\nonumber
\end{eqnarray}
again assuming \eqref{eq:multiRS}.
The fact that the expression above does define a correct probability
law (with a density that integrates to one) is an immediate consequence
of the unbiasedness property given in \eqref{eq:unbiased}.

\begin{prop}
\label{prop:2} The distribution \textup{$\pi_{T}^{N}$}
is such that the marginal distribution of the random variable
$Z_{T}^{\star}\stackrel{\Delta}{=}Z_{T}^{N^{\star}}$
is $\mathbb{Q}_{T}$.
\end{prop}

This proposition is proved in \cite{PMCMC}. To verify this result,
the expectation of functions of $Z_{T}^{\star}$ may be computed by
integrating out the variables in the reverse order $n^\star$, $x_{T}^{1:N}$,
$a_{T}^{1:N},\ldots, x_{1}^{1:N}$, $a_{0}^{1:N}$, $x_{0}^{1:N}$.
We now proceed to state the Gibbs algorithm of \cite{PMCMC}.

Given a sample from $\pi_{T}^{N}$, we can trace the ancestry of the
variable $Z_{T}^{\star}=Z_{T}^{N^{\star}}$ as follows. Let
$B_{t}^{\star}$
for $t\in0\dvtx T$ be the index of the time $t$ ancestor particle of
trajectory $Z_{T}^{\star}$, which is defined recursively backward
as $B_{T}^{\star}=N^{\star}$, then $B_{t}^{\star
}=A_{t}^{B_{t+1}^{\star}}$,
for $t=T-1,\ldots,0$. Finally, let $Z_{t}^{\star}=Z_{t}^{B_{t}^{\star}}$
for $t\in0\dvtx T$, so that $Z_{t}^{\star}$ is precisely the first $t+1$
components of $Z_{T}^{\star}$, that is, $Z_t^\star= [Z_T^\star]_{t+1}$.

Let $Z_{t}^{1:N\setminus\star}$ be the ordered collection of the
$N-1$ trajectories $Z_{t}^{n}$ such that $n\neq B_{t}^{\star}$
(i.e., $n\neq N^{\star}$ when $t=T$), and $Z_{0:T}^{1:N\setminus
\star
}=(Z_{0}^{1:N\setminus\star},\ldots,Z_{T}^{1:N\setminus\star})$.
Define similarly $A_{0:T-1}^{1:N\setminus\star
}=(A_{0}^{1:N\setminus
\star},\ldots,A_{T-1}^{1:N\setminus\star})$,
where $A_{t}^{1:N\setminus\star}$ is $A_{t}^{1:N}$ excluding
$A_{t}^{B_{t+1}^{\star}}$.
It is convenient to apply the following one-to-one transformation
to the argument of $\pi_{T}^{N}$:
\[
\bigl(z_{0:T}^{1:N},a_{0:T-1}^{1:N},n^{\star}
\bigr)\leftrightarrow \bigl(z_{0:T}^{1:N\setminus\star},a_{0:T-1}^{1:N\setminus\star
},z_{0:T}^{\star},b_{0:T-1}^{\star},n^{\star}
\bigr).
\]
With a slight abuse of notation, we identify the law induced by this
transformation (going to the representation with the $b_{t}^{\star}$
variables) as $\pi_{T}^{N}$ as well:
%
\begin{eqnarray}\label{eq:pi_TN_reformulated}
&&\pi_{T}^{N}\bigl(\mathrm{d}z_{0:T}^{1:N\setminus\star},
\mathrm{d}a_{0:T-1}^{1:N\setminus
\star
},\mathrm{d}z_{0:T}^{\star},\mathrm{d}b_{0:T-1}^{\star},\mathrm{d}n^{\star}
\bigr)\nonumber\\
&&\quad=
\frac{1}{N^{T+1}}\bigl(\mathrm{d}b_{0:T-1}^{\star}\,\mathrm{d}n^{\star}\bigr)
\mathbb {Q}_{T}\bigl(\mathrm{d}z_{T}^{\star}\bigr)\prod
_{t=0}^{T-1}\delta_{ (
[z_{T}^{\star
} ]_{t+1} )}
\bigl(\mathrm{d}z_{t}^{\star}\bigr)
\\
&&\qquad{}\times\prod_{n\neq b_{0}^{\star}}m_{0}
\bigl(\mathrm{d}z_{0}^{n}\bigr) \Biggl[\prod
_{t=1}^{T} \prod_{n\neq b_{t}^{\star}}W_{t-1}^{a_{t-1}^n}
\bigl(z_{t-1}^{1:N}\bigr) \,\mathrm{d}a_{t-1}^n
q_{t}\bigl(z_{t-1}^{a_{t-1}^{n}},\mathrm{d}z_{t}^{n}
\bigr) \Biggr].
\nonumber
\end{eqnarray}
Passage from (\ref{eq:defpi}) to (\ref{eq:pi_TN_reformulated}) is
straightforward. It is worth noting that the marginal law of
$(B_{0:T-1}^{\star},N^{\star})$
is the uniform law on the product space $(1\dvtx N)^{T+1}$.

Given a sample $Z_{T}=z_{T}$ from $\mathbb{Q}_{T}$, consider the
following three step sampling procedure that transports $Z_{T}=z_{T}$
to define a new random variable $Z_{T}'\in\mathcal{X}^{T+1}$. Step
1 is to sample the ancestors $(B_{0:T-1}^{\star},N^{\star})$ of the
random variable {$Z_{0:T}^{\star}=z_{0:T}$}
from $\pi_{T}^{N}(\mathrm{d}b_{0:T-1}^{\star},\mathrm{d}n^{\star})$; step~2 is to generate
the $N-1$ remaining trajectories $(Z_{0:T}^{1:N\setminus\star
},A_{0:T-1}^{1:N\setminus\star})$
conditional on the trajectory $(Z_{0:T}^{\star},B_{0:T-1}^{\star
},N^{\star})$
from
\[
\pi_{T}^{N}\bigl(\mathrm{d}z_{0:T}^{1:N\setminus\star},\mathrm{d}a_{0:T-1}^{1:N\setminus
\star
}|z_{0:T}^{\star},b_{0:T-1}^{\star},n^{\star}
\bigr).
\]
(There is no specific difficulty in performing step 2, details to
follow, which is pretty much equivalent to the problem of generating
a particle filter for $T$ time steps.) Note that steps 1 and 2 are
both Gibbs step with respect to \prettyref{eq:pi_TN_reformulated}.
Step 3 is to resample the index $N^{\star}$ from \prettyref{eq:defpi};
hence $N^{\star}$ is sampled from $\pi_{T}^{N}(\mathrm{d}n^{\star
}|z_{0:T}^{1:N},a_{0:T-1}^{1:N})=\mathcal{M}(W_{T}^{1:N}(z_{T}^{1:N}))$,
which is also a Gibbs step, but this time with respect to \prettyref{eq:defpi};
recall that $W_{T}^{n}(z_{T}^{1:N})=G_{T}(z_{T}^{n})/\sum_{m}G_{T}(z_{T}^{m})$.
It follows from Proposition~\ref{prop:2} that the law of
{$Z'_{T}=Z_{T}^{N^{\star}}$}
is also $\mathbb{Q}_{T}$.

Steps 1 to 3 therefore define a Markov kernel $P_{T}^{N}$ that maps
$\setX^{T+1}\rightarrow\mathcal{P}(\setX^{T+1})$ and has $\mathbb
{Q}_T(\mathrm{d}z_{T})$
as its invariant measure. In practice, however,
step 1 is redundant and we may as well set $(b_{0:T-1}^\star,n^\star)$
to the (arbitrary) value $(1,\ldots,1)$
before applying steps 2 and 3, as per the following remark.
%
\begin{remark}
\label{rem:Equivalent} The image of $z_{T}^{\star}\in\mathcal{X}^{T+1}$
under $P_{T}^{N}$ is unchanged
by the choice of $(b_{0:T-1}^{\star},n^{\star})$
for the realization of $(B_{0:T-1}^{\star},N^{\star})$
in the initialization of the CPF kernel.
\end{remark}
This remark follows from the fact that the joint distribution
of $Z_T^{1:N}$ in \eqref{eq:defpi} is exchangeable.
On the other hand, we shall see in
Section~\ref{sec:Backward-sampling} that the equivalent representation
of the particle Gibbs kernel as an update that involves a step that
re-simulates $(B_{0:T-1}^\star,N^\star)$ will be useful to establish
certain properties.

To conclude, and following \cite{PMCMC}, the
CPF kernel may be defined as the succession of the following two steps,
from current value $z_T^\star\in\setX^{T+1}$.

\begin{enumerate}[{CPF-2}]
\item[{CPF-1}] Generate the $N-1$ remaining trajectories of the
particle system
by sampling from the
conditional distribution (deduced from \eqref{eq:pi_TN_reformulated}):
%
\begin{eqnarray}\label{eq:cpfLaw-1}
&&\pi_{T}^{N}\bigl(\mathrm{d}z_{0:T}^{2:N},\mathrm{d}a_{0:T-1}^{2:N}|Z_{0:T}^{1}=z_{0:T}^{\star
},A_{0:T-1}^{1}=
(1,\ldots,1 ),N^{\star}=1\bigr)
\nonumber
\\[-8pt]
\\[-8pt]
\nonumber
&&\quad=m_{0}^{\otimes(N-1)}\bigl(\mathrm{d}z_{0}^{2:N}\bigr)
\prod_{t=1}^{T} \Biggl[\prod
_{n=2}^{N}W_{t-1}^{a_{t-1}^n}
\bigl(z_{t-1}^{1:N}\bigr)\,\mathrm{d}a_{t-1}^n
q_{t}\bigl(z_{t-1}^{a_{t-1}^{n}},\mathrm{d}z_{t}^{n}
\bigr) \Biggr]
\end{eqnarray}
sequentially, that is, sample independently $Z_0^n\sim m_0(\mathrm{d}z_0)$ for
all $n\in2\dvtx N$, then sample independently $A_0^n\sim\mathcal
{M}(W_0^{1:N}(z_0^{1:N}))$ for all $n\in2\dvtx N$, and so on.
(This is equivalent to running a particle algorithm, except that the
trajectory with labels
$(1,\ldots,1)$ is kept fixed.)

\item[{CPF-2}] Sample $N^{\star}$ from $\mathcal
{M}(W_{T}^{1:N}(z_{T}^{1:N}))$,
that is, perform a Gibbs update of $N^\star$ conditional on all the
other variables, relative
to (\ref{eq:pi_TN_reformulated}), and return trajectory $Z^{N^\star}$.
\end{enumerate}

With all these considerations, one sees that the CPF algorithm defines
the following kernel $P_{T}^{N}\dvtx\setX^{T+1}\rightarrow\mathcal
{P}(\setX^{T+1})$:
for $z_{T}^{\star}\in\setX^{T+1}$,
%
\begin{eqnarray}\label{eq:cpfMarkovKer}
&&\bigl(P_{T}^{N}\varphi\bigr) \bigl(z_{T}^{\star}
\bigr) \nonumber\\
&&\quad =\int P_{T}^{N}\bigl(z_{T}^{\star
},\mathrm{d}z_{T}^{\prime}
\bigr)\varphi\bigl(z_{T}^{\prime}\bigr)
\\
&&\quad =\mathbb{E}_{\pi_{T}^{N}} \Biggl\{ \frac{G_{T}(z_{T}^{\star
})}{G_{T}(z_{T}^{\star})+\sum_{m=2}^{N}G_{T}(Z_{T}^{m})}\varphi
\bigl(z_{T}^{\star}\bigr)+\sum_{n=2}^{N}
\frac
{G_{T}(Z_{T}^{n})}{G_{T}(z_{T}^{\star
})+\sum_{m=2}^{N}G_{T}(Z_{T}^{m})}\varphi\bigl(Z_{T}^{n}\bigr) \Biggr
\}.\nonumber
\end{eqnarray}

\section{A coupling of the particle Gibbs Markov kernel}
\label{sec:Main-result}

This section is dedicated to establishing Theorem~\ref{thmm:epsilon}
below. We first make the
following assumption, which is a common assumption to establish the
stability of a Feynman--Kac system (e.g.,~\cite{delMoral:book}).

\renewcommand{\theassumption}{(G)}\label{assG}
\begin{assumption}There exists a sequence of finite positive numbers
$\{g_{t}\}_{t\geq0}$ such that $0<G_{t}(x_{t})\leq g_{t}$ for all
$x_{t}\in\mathcal{X}$, $t\geq0$. Moreover,
\[
\int m_{0}(\mathrm{d}x_{0})G_{0}(x_{0})
\geq\frac{1}{g_{0}},\qquad \inf_{x_{t-1}\in
\setX}\int m_t(x_{t-1},\mathrm{d}x_{t})G_{t}(x_{t})
\geq\frac{1}{g_{t}}, \qquad t>0.
\]
\end{assumption}
Loosely speaking this assumption prevents the reference trajectory of
the particle Gibbs kernel
from dominating the other particles during resampling.

\begin{thmm}
\label{thmm:epsilon}
Under Assumption \textup{\hyperref[assG]{(G)}}, for any $\varepsilon\in(0,1)$
and $T\in\mathbb{N}^{+}$,$ $ there exists $N_{0}\in\mathbb{N}^{+}$,
such that, for all $N\geq N_{0}$, $x_{0:T}$, $\check{x}_{0:T}\in\setX^{T+1}$,
and $\varphi\dvtx\setX^{T+1}\rightarrow[-1,1]$, one has
\[
\bigl\llvert P_{T}^{N}(\varphi) (x_{0:T})-P_{T}^{N}(
\varphi) (\check {x}_{0:T})\bigr\rrvert \leq\varepsilon.
\]
\end{thmm}

The supremum with respect to $\varphi$ of the bounded quantity is the
total variation between the two corresponding distributions (defined
by kernel $P_T^N$ and the two starting points $x_{0:T}$, $\check{x}_{0:T}$).
A~direct corollary of Theorem~\ref{thmm:epsilon} is that, for $N$ large
enough, the kernel $P_{T}^{N}$ is
arbitrarily close to the independent kernel that samples from $\mathbb{Q}_{T}$.
This means that, again for $N$ large enough, the kernel
$P_{T}^{N}$ is uniformly ergodic (see, e.g., \cite{roberts2004general}
for a definition), with an arbitrarily small ergodicity coefficient.

The proof of Theorem~\ref{thmm:epsilon} is based on
coupling: let $\bar{\pi}(\mathrm{d}z_{t}^{\star},\mathrm{d}\check{z}_{t}^{\star})$
be a joint distribution for the couple $(Z_{t}^{\star},\check
{Z}_{t}^{\star})$,
such that the marginal distribution of $Z_{t}^{\star}$, respectively. $\check
{Z}_{t}^{\star}$,
is $P_{T}^{N}(x_{0:T},\mathrm{d}z_{T}^{\star})$, respectively,
$P_{T}^{N}(\check
{x}_{0:T},\mathrm{d}\check{z}_{T}^{\star})$.
Then
\begin{eqnarray*}
P_{T}^{N}(\varphi) (x_{0:T})-P_{T}^{N}(
\varphi) (\check{x}_{0:T}) & =&\E _{\bar{\pi}} \bigl\{ \varphi
\bigl(Z_{T}^{\star}\bigr)-\varphi\bigl(\check
{Z}_{T}^{\star
}\bigr) \bigr\}
\\
& =&\E_{\bar{\pi}} \bigl\{ \bigl(\varphi\bigl(Z_{T}^{\star}
\bigr)-\varphi \bigl(\check {Z}_{T}^{\star}\bigr) \bigr)
\mathbb{I}_{\{Z_{T}\neq\check{Z}_{T}\}
} \bigr\}
\\
& \leq&2\P_{\bar{\pi}} \bigl(Z_{T}^{\star}\neq\check
{Z}_{T}^{\star} \bigr).
\end{eqnarray*}

The following section describes the particular coupling construction we
are using. Section~\ref{sub:Proof-of-inequality} then establishes that
this particular
coupling ensures that
%
\begin{equation}
\P_{\bar{\pi}}\bigl(Z_{T}^{\star}\neq\check{Z}_{T}^{\star}
\bigr)\leq \varepsilon /2\label{eq:proba_epsilon}
\end{equation}
for $N$ large enough, which concludes the proof.

\subsection{Coupling construction}
\label{subsec:couplingConstruct}

The coupling operates on the extended space corresponding to the support
of the conditional distribution \prettyref{eq:cpfLaw-1}. The idea
is to construct two conditional particle systems generated marginally
from \prettyref{eq:cpfLaw-1}, that is, two systems of $N-1$ trajectories,
denoted, respectively, $(Z_{0:T}^{2:N},A_{0:T-1}^{2:N})$ and $(\check
{Z}_{0:T}^{2:N},\check{A}_{0:T-1}^{2:N})$,
that complement, respectively, the trajectory $x_{0:T}$ (first system)
and $\check{x}_{0:T}$ (second system), in such a way that these trajectories
coincide as much as possible. We will denote by $\mathcal
{C}_{t}\subset1\dvtx N$
the set which contains the particle labels $n$ such that $Z_{t}^{n}$
and $\check{Z}_{t}^{n}$ are coupled. Let $\mathcal
{C}_t^{c}=(1\dvtx N)\setminus\mathcal{C}_t$;
by construction, $\mathcal{C}_t^{c}$ always contains $1$, since the frozen
trajectory is relabelled as trajectory $1$ in \prettyref{eq:cpfLaw-1}.
Before we define recursively $\mathcal{C}_t$, $(Z_{0:T}^{2:N},A_{0:T-1}^{2:N})$
and $(\check{Z}_{0:T}^{2:N},\check{A}_{0:T-1}^{2:N})$, we need to
introduce several quantities, such as the following empirical measures,
for $t\geq0$,
\[
\xi_{\mathcal{C}_t}=\sum_{n\in\mathcal{C}_t}\delta
_{Z_{t}^{n}}(\mathrm{d}z_{t}), \qquad \xi_{\mathcal{C}_t
^{c}}=\sum
_{n\in\mathcal{C}_t^{c}}\delta_{Z_{t}^{n}}(\mathrm{d}z_{t}),\qquad \check{
\xi}_{\mathcal{C}_t
^{c}}=\sum_{n\in\mathcal{C}_t^{c}}\delta_{\check{Z}_{t}^{n}}(\mathrm{d}z_{t}),
\]
the following probability measures, $\mu_{0}(\mathrm{d}z_{0})=m_{0}(\mathrm{d}z_{0})$,
and for $t\geq1$,
\begin{eqnarray*}
\mu_{t}(\mathrm{d}z_{t}) & =&\int_{\setX^{t}}
\Psi_{G_{t-1}}(\xi_{\mathcal{C}_{t-1}
}) (\mathrm{d}z_{t-1})q_{t}(z_{t-1},\mathrm{d}z_{t}),
\\
\mu_{t}^{c}(\mathrm{d}z_{t}) & =&\int
_{\setX^{t}}\Psi_{G_{t-1}}(\xi _{\mathcal{C}_{t-1}
^{c}})
(\mathrm{d}z_{t-1})q_{t}(z_{t-1},\mathrm{d}z_{t}),
\\
\check{\mu}_{t}^{c}(\mathrm{d}\check{z}_{t}) & =&\int
_{\setX^{t}}\Psi _{G_{t-1}}(\check{\xi}_{\mathcal{C}_{t-1}^{c}}) (\mathrm{d}
\check {z}_{t-1})q_{t}(z_{t-1},\mathrm{d}z_{t}),
\end{eqnarray*}
where
\[
\Psi_{G_{t-1}}(\xi_{\mathcal{C}_{t-1}}) (\mathrm{d}z_{t-1})=
\frac{\xi
_{\mathcal{C}_{t-1}
}(\mathrm{d}z_{t-1})G_{t-1}(z_{t-1})}{\int_{\setX^{t}}\xi_{\mathcal{C}_{t-1}
}(\mathrm{d}z_{t-1})G_{t-1}(z_{t-1})},
\]
the measures $\Psi_{G_{t-1}}(\xi_{\mathcal{C}_{t-1}^{c}})(\mathrm{d}z_{t-1})$
and $\Psi
_{G_{t-1}}(\check{\xi}_{\mathcal{C}_{t-1}^{c}})(\mathrm{d}z_{t-1})$
being defined similarly, and finally the constants
\[
\lambda_{t-1}=\frac{\xi_{\mathcal{C}_{t-1}}(G_{t-1})}{\xi
_{\mathcal{C}_{t-1}}(G_{t-1})+\xi
_{\mathcal{C}_{t-1}
^{c}}(G_{t-1})},\qquad \check{\lambda}_{t-1}=
\frac{\xi_{\mathcal{C}_{t-1}
}(G_{t-1})}{\xi_{\mathcal{C}_{t-1}}(G_{t-1})+\check{\xi}_{\mathcal
{C}_{t-1}^{c}}(G_{t-1})},
\]
and the measures
\begin{eqnarray*}
\nu_{t} & =&\frac{|\lambda_{t-1}-\check{\lambda}_{t-1}|}{1-\lambda
_{t-1}\wedge\check{\lambda}_{t-1}}\mu_{t}+\frac{1-\lambda
_{t-1}\vee
\check{\lambda}_{t-1}}{1-\lambda_{t-1}\wedge\check{\lambda
}_{t-1}}\mu
_{t}^{c},
\\
\check{\nu}_{t} & =&\frac{|\lambda_{t-1}-\check{\lambda
}_{t-1}|}{1-\lambda_{t-1}\wedge\check{\lambda}_{t-1}}\mu_{t}+
\frac
{1-\lambda_{t-1}\vee\check{\lambda}_{t-1}}{1-\lambda_{t-1}\wedge
\check
{\lambda}_{t-1}}\check{\mu}_{t}^{c},
\\
\kappa_{t}(\mathrm{d}z_{t},\mathrm{d}\check{z}_{t}) & =&
\nu_{t}(\mathrm{d}z_{t})\check{\mu }_{t}^{c}(\mathrm{d}
\check{z}_{t})\mathbb{I}_{\{\lambda_{t-1}>\check{\lambda
}_{t-1}\}}+\mu_{t}^{c}(\mathrm{d}z_{t})
\check{\nu}_{t}(\mathrm{d}\check {z}_{t})\mathbb {I}_{\{\check{\lambda}_{t-1}\geq\lambda_{t-1}\}}.
\end{eqnarray*}

We now construct $\mathcal{C}_t$ and the two particle systems as
follows. First,
set $\mathcal{C}_{0}=2\dvtx N$, hence $\mathcal{C}_{0}^{c}=\{1\}$, draw
$Z_{0}^{n}$ independently from $m_{0}$, and set $\check{Z}_{0}^{n}=Z_{0}^{n}$,
for all $n\in2\dvtx N$. Recall that $Z_{0}^{1}=x_{0}$ and $\check
{Z}_{0}^{1}=\check{x}_{0}$.

To progress from time $t-1\geq0$ to time $t$, we note that there
is a $\lambda_{t-1}$ (resp., $\check{\lambda}_{t-1}$) probability
that $A_{t-1}^{n}$ (resp., $\check{A}_{t-1}^{n}$) is drawn from
$\mathcal{C}_{t-1}$,
for any $n\in2\dvtx N$. Hence, the maximum coupling probability for $
(A_{t-1}^{n},\check{A}_{t-1}^{n} )$
is $\lambda_{t-1}\wedge\check{\lambda}_{t-1}$. Thus, with probability
$\lambda_{t-1}\wedge\check{\lambda}_{t-1}$, we sample $A_{t-1}^{n}$
from $\mathcal{C}_{t-1}$ (with probability proportional to
$G_{t-1}(Z_{t-1}^{m})$,
for $m\in\mathcal{C}_{t-1}$), $Z_{t}^{n}\sim
q_{t}(Z_{t-1}^{A_{t-1}^{n}},\mathrm{d}Z_{t})$,
and take $ (\check{A}_{t-1}^{n},\check{Z}_{t}^{n} )=
(A_{t-1}^{n},Z_{t}^{n} )$.
Marginally, $Z_{t}^{n}=\check{Z}_{t}^{n}$ is drawn from $\mu_{t}$,
and we set $n\in\mathcal{C}_t$.

Conditional on not being coupled (hence, we set $n\in\mathcal{C}_t^{c})$,
$(A_{t-1}^{n},Z_{t}^{n})$
and $(\check{A}_{t-1}^{n},\check{Z}_{t-1}^{n})$ may be sampled independently
using the same ideas. Assume $\lambda_{t-1}\leq\check{\lambda}_{t-1}$.
With probability $(\check{\lambda}_{t-1}-\lambda_{t-1})$, one should
sample $A_{t-1}^{n}$ from $\mathcal{C}_{t-1}^{c}$ and $\check{A}_{t-1}^{n}$
from $\mathcal{C}_{t-1}$. And with probability $1-\lambda_{t-1}\vee
\check
{\lambda}_{t-1}$,
both $A_{t}^{n}$ and $\check{A}_{t}^{n}$ may be sampled from $\mathcal
{C}_t^{c}$.
Either way, $Z_{t}^{n}\sim q_{t}(Z_{t-1}^{A_{t-1}^{n}},\mathrm{d}Z_{t})$,
$\check{Z}_{t}^{n}\sim q_{t}(\check{Z}_{t-1}^{\check{A}_{t-1}^{n}},\mathrm{d}Z_{t})$,
independently. By symmetry, the case $\lambda_{t-1}\geq\check
{\lambda}_{t-1}$
works along the same lines. Marginally (when integrating out $A_{t-1}^{n}$
and $\check{A}_{t-1}^{n})$, and conditional on not being coupled,
the pair $ (Z_{t}^{n},\check{Z}_{t}^{n} )$ is drawn from
$\kappa_{t}$. Clearly, this construction maintains the correct marginal
distribution for the two particle systems.

At the final time $T$, the trajectories $Z_{T}^{\star}$, $\check
{Z}_{T}^{\star}$
that are eventually selected, that is, the output of Markov kernels
$P_{T}^{N}(x_{0:T},\mathrm{d}z_{T}^{\star})$ and $P_{T}^{N}(\check
{x}_{0:T},\mathrm{d}\check{z}_{T}^{\star})$
may be coupled exactly in the same way: with probability $\lambda
_{T}\wedge\check{\lambda}_{T}$,
they are taken to be equal, and $Z_{T}=Z_{T}^{\star}$ is sampled
from $\mu_{T}$; and with probability $(1-\lambda_{T}\wedge\check
{\lambda}_{T})$,
$(Z_{T}^{\star},\check{Z}_{T}^{\star})$ is sampled from $\kappa_{T}$.

The motivation for this coupling construction is that it is the \textit{maximal
coupling} (see \cite{Lin1992} for a definition) for quantifying the
total variation norm between CPF kernels
{$P_{T}^{N}(x_{0:T},\cdot)$} and $P_{T}^{N}(\check{x}_{0:T},\cdot)$
when either $T=0$ or $T>0$ and $m_{t}$ is a Dirac measure for all
$t$. Details of proof of this fact can be obtained from the authors.

\subsection{\texorpdfstring{Proof of inequality \protect\eqref{eq:proba_epsilon}}{Proof of inequality (9)}}
\label{sub:Proof-of-inequality}

We now prove that the coupling construction described in the previous
section is such that inequality \eqref{eq:proba_epsilon} holds for
$N$ large enough.

By construction, one has that $\P(Z_{T}^{\star}=\check{Z}_{T}^{\star
})\geq\E(\lambda_{T}\wedge\check{\lambda}_{T})$.
Consider the event
\[
\ev= \biggl\{ \frac{\xi_{\mathcal{C}_T}(G_{T})}{\llvert \mathcal
{C}_T\rrvert }\geq\frac
{\mu
_{T}(G_{T})}{2} \biggr\}.
\]
Given Assumption \hyperref[assG]{(G)}, and the definition of $\lambda_{T}$, one has
\[
1-\lambda_{T}\leq\frac{g_{T}\llvert \mathcal{C}_T^{c}\rrvert }{\xi
_{\mathcal{C}_T}(G_{T})}
\]
and the same inequality holds for $1-\check{\lambda}_{T}$, which
leads to
\[
(1-\lambda_{T}\wedge\check{\lambda}_{T})\times
\mathbb{I}_{\ev
}\leq2\frac
{\llvert \mathcal{C}_T^{c}\rrvert g_{T}}{|\mathcal{C}_T|\mu
_{T}(G_{T})}\times\mathbb {I}_{\ev
}
\leq2\frac{\llvert \mathcal{C}_T^{c}\rrvert g_{T}^{2}}{|\mathcal
{C}_T|}\times\mathbb {I}_{\ev},
\]
where the second inequality is due to Assumption \hyperref[assG]{(G)}. Therefore, for
$k_{1},\ldots,k_{T}\in1:N$,
\begin{eqnarray*}
&&\mathbb{E} \bigl\{ (\lambda_{T}\wedge\check{\lambda }_{T}
)\times\mathbb{I}_{\ev}\vert\llvert \C_{1}\rrvert
=N-k_{1},\ldots,\llvert \mathcal{C}_T\rrvert
=N-k_{T} \bigr\}
\\
&&\quad\geq \biggl(1-2\frac{k_{T}}{N-k_{T}}g_{T}^{2}
\biggr)^{+}\mathbb {E} \bigl\{ \mathbb{I}_{\ev}\vert\llvert
\C_{1}\rrvert =N-k_{1},\ldots,\llvert \mathcal{C}_T
\rrvert =N-k_{T} \bigr\}.
\end{eqnarray*}
Conditional on $Z_{T-1}^{1:N}$, and $n\in\mathcal{C}_T$, $Z_{T}^{n}$
is an
independent draw from $\mu_{T}(\mathrm{d}z_{T})$. Thus, in order
to lower bound the probability of event ${\mathcal{A}}$,
we may apply Hoeffding's inequality \cite{hoeffding1963probability}
to the empirical mean $\xi_{\mathcal{C}_T}(G_{T})/\llvert \mathcal
{C}_T\rrvert $ as follows.
Again per Assumption \hyperref[assG]{(G)}, noting that
$0< G_T(z_T^n)\leq g_T$ and the one-step predicted potential is bounded
below uniformly by $g_T^{-1}$,
\begin{eqnarray*}
&& \mathbb{E} \bigl\{ (1-\mathbb{I}_{\ev} )\rrvert \llvert
\C_{1}\vert=N-k_{1},\ldots,\llvert \C_{T}\rrvert
=N-k_{T},Z_{T-1}^{1:N}=z_{T-1}^{1:N}
\bigr\}
\\
& &\quad= \P \biggl( \frac{\xi_{\mathcal{C}_T}(G_{T})}{\llvert \mathcal
{C}_T\rrvert }< \frac{\mu_{T}(G_{T})}{2} \vert \llvert
\C_{1}\rrvert =N-k_{1},\ldots,\llvert \C _{T}
\rrvert =N-k_{T},Z_{T-1}^{1:N}=z_{T-1}^{1:N}
\biggr)
\\
&&\quad \leq\exp \biggl(-2(N-k_{T})\frac{\mu
_{T}(G_{T})^{2}}{4g_{T}^{2}} \biggr)
\\
& &\quad\leq\exp \biggl(-\frac{(N-k_{T})}{2g_{T}^{4}} \biggr).
\end{eqnarray*}
Thus,
\begin{eqnarray*}
&&\mathbb{E} \bigl\{ (\lambda_{T}\wedge\check{\lambda}_{T})
\times \mathbb{I}_{\ev}\vert\llvert C_{1}\rrvert
=N-k_{1},\ldots ,\llvert C_{T}\rrvert =N-k_{T}
\bigr\}
\\
&&\quad\geq \biggl(1-2\frac{k_{T}}{N-k_{T}}g_{T}^{2}
\biggr)^{+} \biggl\{ 1-\exp \biggl(-\frac{(N-k_{T})}{2g_{T}^{4}} \biggr) \biggr
\}.
\end{eqnarray*}
Finally, for any sequence of integers $L_{1:T}\in(1:N)^{T}$,
\begin{eqnarray*}
\mathbb{E}(\lambda_{T}\wedge\check{\lambda}_{T}) & \geq&
\sum_{k_{1}=1}^{L_{1}}\cdots\sum
_{k_{T}=1}^{L_{T}}\E \bigl\{ (\lambda_{T}\wedge
\check{\lambda}_{T} ) \times\mathbb{I}_{\mathcal{A}}\vert\llvert
C_{1}\rrvert =N-k_{1},\ldots,\llvert C_{T}
\rrvert =N-k_{T} \bigr\}
\\
&&{} \times\P \bigl(\llvert \C_{1}\rrvert =N-k_{1},\ldots,
\llvert \mathcal{C}_T\rrvert =N-k_{T} \bigr)
\\
& \geq& \biggl(1-2\frac{L_{T}}{N-L_{T}}g_{T}^{2} \biggr)
\biggl\{ 1-\exp \biggl(-\frac{(N-L_{T})}{2g_{T}^{4}} \biggr) \biggr\}
\\
&&{} \times\sum_{k_{1}=1}^{L_{1}}\cdots\sum
_{k_{T}=1}^{L_{T}}\P \bigl(\llvert \C_{1}\rrvert
=N-k_{1},\ldots ,\llvert \mathcal{C}_T\rrvert
=N-k_{T} \bigr)
\end{eqnarray*}
provided $N$ is large enough. To conclude, we resort to a technical
lemma, proven in the following section, that states it is possible
to choose $L_{1},\ldots,L_{T}$ large enough so as to make the sum
of probabilities in the last line above as large as needed. In addition,
for $L_{T}$ fixed and $N$ large enough, the two factors in front
of that sum are arbitrarily close to one.

\subsection{Technical lemma}
%
\begin{lem}
\label{lem:probC1toCt_lowerBound}Under Assumption \textup{\hyperref[assG]{(G)}},
and for any $\delta\in(0,1)$, $T\in\mathbb{N}^{+}$, there exist
positive integers $N_{0}$, $L_{1},\ldots,L_{T}$ such that for any
$N\geq N_{0}$ and $x_{0:T}$, $\check{x}_{0:T}\in\setX^{T+1}$,
\[
\sum_{k_{1}=1}^{L_{1}}\cdots\sum
_{k_{T}=1}^{L_{T}}\P \bigl(\llvert \C _{1}\rrvert
=N-k_{1},\ldots,\llvert \mathcal{C}_T\rrvert
=N-k_{T} \bigr)\geq(1-\delta)^{3T}.
\]
\end{lem}

\begin{pf} Let
$\omega_{t}=\lambda_{t}\wedge\check{\lambda}_{t}$
and recall that $\omega_{t}$ is the probability (conditional on $Z_{t}^{1:N})$
that $n\in\mathcal{C}_{t+1}$, that is, that particles $Z_{t+1}^{n}$
and $\check
{Z}_{t+1}^{n}$
are coupled. Thus, and using the fact that
the particle system is exchangeable, one has
%
\begin{eqnarray}
\label{eq:PCk} &&\P \bigl(\llvert \C_{1}\rrvert
=N-k_{1},\ldots ,\llvert \mathcal{C}_T\rrvert
=N-k_{T} \bigr)
\nonumber
\\
&&\quad = \Biggl\{ \prod_{t=1}^{T}\pmatrix{N-1 \cr
N-k_{t}} \Biggr\} \P \bigl(\C _{1}^{c}=
(1\dvtx k_{1} ),\ldots,\mathcal{C}_T^{c}=
(1\dvtx k_{T} ) \bigr)
\nonumber
\\[-8pt]
\\[-8pt]
\nonumber
&&\quad =\int\prod_{t=0}^{T-1}\pmatrix{N-1 \cr
N-k_{t+1}} \Biggl\{ \prod_{n=2}^{k_{t}}
\kappa_{t}\bigl(\mathrm{d}z_{t}^{n},\mathrm{d}\check{z}_{t}^{n}
\bigr)\prod_{n=k_{t}+1}^{N}\mu_{t}
\bigl(\mathrm{d}z_{t}^{n}\bigr) \Biggr\} (1-\omega
_{t})^{(k_{t+1}-1)}\omega_{t}^{(N-k_{t+1})}
\\
&&\quad \geq\int\prod_{t=0}^{T-1} \Biggl\{ \prod
_{n=2}^{k_{t}}\kappa _{t}
\bigl(\mathrm{d}z_{t}^{n},\mathrm{d}\check{z}_{t}^{n}
\bigr)\prod_{n=k_{t}+1}^{N}\mu _{t}
\bigl(\mathrm{d}z_{t}^{n}\bigr) \Biggr\} \biggl\{ \frac{(N-1)!(1-\omega
_{t})^{(k_{t+1}-1)}}{(N-k_{t+1})!}
\biggr\} \biggl\{ \frac{\omega
_{t}^{(N-k_{t+1})}}{(k_{t+1}-1)!} \biggr\}\mathbb{I}_{\ev_{t}}\nonumber
\end{eqnarray}
with the convention that $k_{0}=1$, that $\mu_{0}(\mathrm{d}z_{0})=m_{0}(\mathrm{d}z_{0})$,
and that empty products equal one, and defining the event $\mathcal{A}_{t}$
as
%
\begin{equation}
\ev_{t}= \biggl\{ \frac{\xi_{\mathcal{C}_t}(G_{t})}{|\mathcal
{C}_t|}\geq\frac{\mu
_{t}(G_{t})}{2} \biggr
\}.\label{eq:Fevents}
\end{equation}
Note that the two integrals above are with respect to a joint distribution
which corresponds to a chain rule decomposition that works forward in time:
$\kappa_1$ and $\mu_1$ are distributions conditional on $Z_0^{1:N}$ and
so on.
In addition, these chained conditional distributions are
such that $\llvert \mathcal{C}_t\rrvert =N-k_{t}$ with probability one.\

For the sake of transparency, we complete the proof for $T=2$, but
we note that exactly the same steps employed may be extended to the general
case where $T>2$.

The key idea is to replace the two factors in the integrand of \eqref{eq:PCk}
with their large $N$ values which we now define. Let
%
\begin{equation}
\Lambda_{t}=|\mathcal{C}_t|\times\frac{\xi_{\mathcal
{C}_t^{c}}(G_{t})\vee\check{\xi}_{\mathcal{C}_t
^{c}}(G_{t})}{\xi_{\mathcal{C}_t}(G_{t})}\qquad
\mbox{for }|\mathcal {C}_t|>0\label{eq:capLamda}
\end{equation}
and set $\Lambda_{t}=0$ if $|\mathcal{C}_t|=0$. Using Lemma~\ref
{lem:BloodyLemma},
stated and proved at the end of this section, one has, for fixed $k_{1}$,
$k_{2}$, $\delta$, and $N$ large enough, that the integral in \eqref{eq:PCk}
is larger than
%
\begin{eqnarray}\label{eq:replace_with_largeNvalues}
&&(1-\delta)^{2}\int_{\ev_{0}} \Biggl\{ \prod
_{n=2}^{N}\mu _{0}\bigl(\mathrm{d}z_{0}^{n}
\bigr) \Biggr\} \frac{\mathrm{e}^{-\Lambda_{0}}\Lambda
_{0}^{k_{1}-1}}{(k_{1}-1)!}
\nonumber
\\[-8pt]
\\[-8pt]
\nonumber
&&\quad{}\times \int_{\ev_{1}} \Biggl\{ \prod
_{n=2}^{k_{1}}\kappa_{1}
\bigl(\mathrm{d}z_{1}^{n},\mathrm{d}\check{z}_{1}^{n}
\bigr)\prod_{n=k_{1}+1}^{N}\mu _{1}
\bigl(\mathrm{d}z_{1}^{n}\bigr) \Biggr\} \frac{\mathrm{e}^{-\Lambda_{1}}\Lambda
_{1}^{k_{2}-1}}{(k_{2}-1)!}.
\end{eqnarray}
We now explain how to choose $L_{1}$, $L_{2}$ such that, for $N$
large enough,
%
\begin{eqnarray}\label
{eq:howtochooseL1L2}
&&\int_{\ev_{0}} \Biggl\{ \prod_{n=2}^{N}
\mu_{0}\bigl(\mathrm{d}z_{0}^{n}\bigr) \Biggr\} \sum
_{k_{1}=1}^{L_{1}}\frac{\mathrm{e}^{-\Lambda_{0}}\Lambda
_{0}^{k_{1}-1}}{(k_{1}-1)!}
\nonumber
\\[-8pt]
\\[-8pt]
\nonumber
&&\quad{}\times\int_{\ev_{1}} \Biggl\{ \prod
_{n=2}^{k_{1}}\kappa _{1}
\bigl(\mathrm{d}z_{1}^{n},\mathrm{d}\check{z}_{1}^{n}
\bigr)\prod_{n=k_{1}+1}^{N}\mu _{1}
\bigl(\mathrm{d}z_{1}^{n}\bigr) \Biggr\} \sum
_{k_{2}=1}^{L_{2}}\frac{\mathrm{e}^{-\Lambda
_{1}}\Lambda_{1}^{k_{2}-1}}{(k_{2}-1)!}\geq(1-
\delta)^{4}.
\end{eqnarray}

First note that, given Assumption \hyperref[assG]{(G)}, and since $\llvert \C_{0}\rrvert =N-1$,
$\llvert \C_{1}\rrvert =N-k_{1}$ (with probability one under the conditional
distribution that appears in \eqref{eq:PCk}, for $t=1$, as explained
above), and since $\xi_{\mathcal{C}_t}(G_{t})\geq\llvert \mathcal
{C}_t\rrvert \mu_{t}(G_{t})/2$
(by event $\mathcal{A}_t$), one has (again with probability one under the
same conditional distribution):
%
\begin{equation}
0\leq\Lambda_{0}\times\mathbb{I}_{\mathcal{A}_{0}}\leq
2g_{0}^{2}\times \mathbb{I}_{\mathcal{A}_{0}},\qquad 0\leq
\Lambda_{1}\times\mathbb {I}_{\mathcal{A}_{1}}\leq2g_{1}^{2}k_{1}
\times\mathbb{I}_{\mathcal
{A}_{1}}\label{eq:bound_Lambdat}
\end{equation}
for $N>k_{1}$ (otherwise the probability that $
\llvert \C_{1}\rrvert =N-k_{1}$
would be zero). Choose $L_{1}$, then $L_{2}$, such that
\begin{eqnarray*}
\sum_{k_{1}=1}^{L_{1}}\frac
{\mathrm{e}^{-2g_{0}^{2}}(2g_{0}^{2})^{k_{1}-1}}{(k_{1}-1)!}&\geq&1-
\delta, \\
\sum_{k_{2}=1}^{L_{2}}\frac
{\mathrm{e}^{-(2g_{1}^{2}L_{1})}(2g_{1}^{2}L_{1})^{k_{2}-1}}{(k_{2}-1)!}
&\geq& 1-\delta.
\end{eqnarray*}
Since $x\rightarrow \mathrm{e}^{-x}\sum_{k=0}^{L}x^{k}/k!$ is a decreasing
function for $x>0$, this choice of $L_{2}$ ensures that
\begin{eqnarray*}
&&\int_{\ev_{1}} \Biggl\{ \prod_{n=2}^{k_{1}}
\kappa _{1}\bigl(\mathrm{d}z_{1}^{n},\mathrm{d}\check
{z}_{1}^{n}\bigr)\prod_{n=k_{1}+1}^{N}
\mu_{1}\bigl(\mathrm{d}z_{1}^{n}\bigr) \Biggr\} \sum
_{k_{2}=1}^{L_{2}}\frac{\mathrm{e}^{-\Lambda_{1}}\Lambda
_{1}^{k_{2}-1}}{(k_{2}-1)!}
\\
&&\quad\geq(1-\delta)\int_{\ev_{1}} \Biggl\{ \prod
_{n=2}^{k_{1}}\kappa _{1}
\bigl(\mathrm{d}z_{1}^{n},\mathrm{d}\check{z}_{1}^{n}
\bigr)\prod_{n=k_{1}+1}^{N}\mu _{1}
\bigl(\mathrm{d}z_{1}^{n}\bigr) \Biggr\}.
\end{eqnarray*}

By Hoeffding's inequality (in the same way as in the previous section),
\begin{eqnarray*}
\int_{\ev_{1}^{c}} \Biggl\{ \prod_{n=k_{1}+1}^{N}
\mu _{1}\bigl(\mathrm{d}z_{1}^{n}\bigr) \Biggr\} &\leq&\exp
\biggl\{-2(N-k_{1})\frac{\mu
_{1}(G_{1})^{2}}{4g_{1}^{2}} \biggr\} \\
&\leq&\exp \biggl
\{-(N-k_{1})\frac{g_{1}^{-4}}{2} \biggr\},
\end{eqnarray*}
where the last inequality follows from Assumption \hyperref[assG]{(G)}. Using the same
calculations for the first integral in \eqref{eq:howtochooseL1L2}
(i.e., applying Hoeffding's inequality to $\ev_{0}$ again in the
same way),
one obtains eventually
\begin{eqnarray*}
&&\int_{\ev_{0}} \Biggl\{ \prod
_{n=2}^{N}\mu_{0}\bigl(\mathrm{d}z_{0}^{n}
\bigr) \Biggr\} \sum_{k_{1}=1}^{L_{1}}
\frac{\mathrm{e}^{-\Lambda_{0}}\Lambda
_{0}^{k_{1}-1}}{(k_{1}-1)!}\\
&&\qquad{} \times\int_{\ev_{1}} \Biggl\{ \prod
_{n=2}^{k_{1}}\kappa_{1}\bigl(\mathrm{d}z_{1}^{n},\mathrm{d}
\check{z}_{1}^{n}\bigr)\prod_{n=k_{1}+1}^{N}
\mu_{1}\bigl(\mathrm{d}z_{1}^{n}\bigr) \Biggr\} \sum
_{k_{2}=1}^{L_{2}}\frac
{\mathrm{e}^{-\Lambda_{1}}\Lambda_{1}^{k_{2}-1}}{(k_{2}-1)!}
\\
&&\quad \geq(1-\delta)^{2} \biggl[1-\exp \biggl\{-(N-L_{1})
\frac{g_{1}^{-4}}{2} \biggr\} \biggr] \biggl[1-\exp \biggl\{-(N-1)\frac{g_{0}^{-4}}{2}
\biggr\} \biggr]
\\
&&\quad \geq(1-\delta)^4
\end{eqnarray*}
for $N$ large enough and, therefore,
combining this with \eqref{eq:replace_with_largeNvalues},
one may conclude that
\[
\sum_{k_{1}=1}^{L_{1}}\sum
_{k_{2}=1}^{L_{2}}\P\bigl(\llvert \C _{1}\rrvert
=N-k_{1},\llvert \C_{2}\rrvert =N-k_{2}\bigr)
\geq(1-\delta)^{6}
\]
provided $N$ is taken to be large enough.
\end{pf}

To conclude the proof, we state and prove the following lemma, which
we used in the proof above in order to replace the two last factors
in \eqref{eq:PCk} by their large-$N$ values.

\begin{lem}
\label{lem:BloodyLemma}Assume \textup{\hyperref[assG]{(G)}}. For any given $\delta>0$ and
positive integers $k_{1},\ldots,k_{T}$, there exists a positive integer
$N_{0}$ such that the following inequalities hold for all $N\geq N_{0}$ and $x_{0:T}$, $\check{x}_{0:T}\in\setX^{T+1}$:
\begin{eqnarray*}
\omega_{t}^{(N-k_{t+1})}\exp(\Lambda_{t})\times\mathbb
{I}_{\mathcal
{A}_{t}} & \geq&(1-\delta)\times\mathbb{I}_{\mathcal{A}_{t}},
\\
\frac{(N-1)!}{(N-k_{t+1})!} (1-\omega_{t} )^{(k_{t+1}-1)}\frac
{1}{\Lambda_{t}^{k_{t+1}-1}}
\times\mathbb{I}_{\mathcal{A}_{t}} & \geq& (1-\delta)\times\mathbb{I}_{\mathcal{A}_{t}}.
\end{eqnarray*}
\end{lem}
\begin{pf}
Given that $\llvert \mathcal{C}_t\rrvert =N-k_{t}$ and the respective
definitions
of $\omega_{t}$ and $\Lambda_{t}$, one has
\[
\frac{1-\omega_{t}}{\omega_{t}}=\frac{\Lambda_{t}}{N-k_{t}}
\]
and, therefore, for $N$ large enough and $k_{t}$ fixed, and conditional
on $\mathbb{I}_{\mathcal{A}_{t}}=1$, the probability $1-\omega_{t}$
may be made arbitrarily small, given that $\Lambda_{t}\mathbb
{I}_{\mathcal{A}_{t}}$
is a bounded quantity; see \eqref{eq:bound_Lambdat}. Since $\log
(1+x)\geq x-x^{2}$
for $x\geq-1/2$, one has, for $N$ large enough (so that $x=\omega
_{t}-1\geq-1/2$),
and conditional on $\mathbb{I}_{\mathcal{A}_{t}}=1$,
\[
\Lambda_{t}+(N-k_{t+1})\log\omega_{t}\geq
\Lambda_{t} \biggl\{ 1-\frac
{N-k_{t+1}}{N-k_{t}}\omega_{t}-
\frac{N-k_{t+1}}{ (N-k_{t}
)^{2}}\Lambda_{t}\omega_{t}^{2}
\biggr\},
\]
which can be clearly made arbitrarily small (in absolute value) by taking
$N$ large enough, since both $\omega_{t}$ and $\Lambda_{t}$ are
bounded quantities. The second inequality may be proved along the
same lines.
\end{pf}

\section{Backward sampling}
\label{sec:Backward-sampling}

This section discusses the backward sampling (BS) step proposed by
\cite{Whiteley_disc_PMCMC}
so as to improve the mixing of particle Gibbs. It is convenient in this
section to revert to standard notation based
on the initial process $X_{t}$, rather than on notation based on
trajectories $Z_{t}=X_{0:t}$. Thus, we now consider the following
(extended) invariant distribution for the CPF kernel
%
\begin{eqnarray}\label
{eq:defpi_xt_notations}
&&\pi_{T}^{N} \bigl(\mathrm{d}x_{0:T}^{1:N},\mathrm{d}a_{0:T-1}^{1:N},\mathrm{d}n^\star
\bigr)\nonumber\\
 &&\quad =\frac
{1}{\mathcal{Z}_T}m_{0}^{\otimes N}\bigl(\mathrm{d}x_{0}^{1:N}
\bigr)
\nonumber
\\[-8pt]
\\[-8pt]
\nonumber
& &\qquad{}\times\prod_{t=1}^{T}
\Biggl\{ \Biggl[\frac{1}{N}\sum_{n=1}^{N}G_{t-1}
\bigl(x_{t-1}^{n}\bigr) \Biggr] \prod
_{n=1}^{N} \bigl[ W_{t-1}^{a_{t-1}^n}
\bigl(x_{t-1}^{1:N}\bigr)\,\mathrm{d}a_{t-1}^n
m_{t}\bigl(x_{t-1}^{a_{t-1}^{n}},\mathrm{d}x_{t}^{n}
\bigr) \bigr] \Biggr\}
\\
& &\qquad{}\times\frac{1}{N}G_{T}\bigl(x_{T}^{n^\star}
\bigr),
\nonumber
\end{eqnarray}
where $W_t^n(x_t^{1:N})=G_t(x_t^n)/\sum_{m=1}^N G_t(x_t^m)$;
compared to \prettyref{eq:defpi}, this equation represents a simple
change of variables.

In this new set of notation, the $n$th trajectory $Z_{T}^{n}$ becomes
a deterministic function of the particle system
$(X_{0:T}^{1:N},A_{0:T-1}^{1:N})$,
that may be defined as follows: $Z_{T}^{n}=(X_{0}^{B_{0}^{n}},\ldots
,X_{T}^{B_{T}^{n}})$,
where the indexes $B_{T}^{n}$'s are defined recursively as: $B_{T}^{n}=n$,
$B_{t}^{n}=A_{t}^{B_{t+1}^{n}}$, for $t\in0\dvtx (T-1)$. Similarly, as
noted before, $Z_{T}^{\star}=Z_{T}^{N^{\star}}=(X_{0}^{B_{0}^{\star
}},\ldots,X_{T}^{B_{T}^{\star}})$,
with $B_{T}^{\star}=N^{\star}$, $B_{t}^{\star}=A_{t}^{B_{t+1}^{\star}}$,
that is, $Z_{T}^{\star}$ is a deterministic function of
$(X_{0:T}^{1:N},A_{0:T-1}^{1:N},N^{\star})$.

Whiteley \cite{Whiteley_disc_PMCMC}, in his discussion of \cite{PMCMC}
(see also \cite{Lindsten2012}), suggested to add the following BS
(backward sampling) step to a particle Gibbs update.
\begin{enumerate}[{CPF-3}]
\item[{CPF-3}] Let $B_{T}^{\star}=N^{\star}$, then, recursively for $t=T-1$
to $t=0$, sample index $B_{t}^{\star}=A_{t}^{B_{t+1}^{\star}}\in1\dvtx N$,
conditionally on $B_{t+1}^{\star}=b$, from the distribution
%
\begin{eqnarray}\label{eq:BS_distr}
&&\pi_{T}^{N} \bigl(A_{t}^{b}=a_{t}^{b}|X_{0:T}^{1:N}=x_{0:T}^{1:N},
A_{t}^{-b}=a_{t}^{-b},
A_{0:t-1}^{1:N}=a_{0:t-1}^{1:N},
A_{t+1:T}^{1:N}=a_{t+1:T}^{1:N},N^{\star}=n^{\star}
\bigr)
\nonumber
\\[-3pt]
\\[-6pt]
\nonumber
&&\quad\propto W_t^{a_t^b} m_{t+1}\bigl(x_{t}^{a_{t}^{b}},x_{t+1}^{b}
\bigr)
\end{eqnarray}
and set $X_{t}^{\star}=x_{t}^{b_{t}^{\star}}$. (Recall that
$Z_{T}^{\star}=(X_{0}^{B_{0}^{\star}},\ldots,X_{T}^{B_{T}^{\star}})$.)
\end{enumerate}
In \prettyref{eq:BS_distr}, $A_{t}^{-b}$ is $A_{t}^{1:N}$ minus
$A_{t}^{b}$, $a_{t}^{-b}$ is defined similarly, and
$m_{t+1}(x_{t}^{a_{t}^{b}},x_{t+1}^{b})$ is the probability
density (relative to measure $\mathrm{d}x_{t+1}$) of conditional distribution
$m_{t+1}(x_{t}^{a_{t}^{b}},\mathrm{d}x_{t+1})$ evaluated at point $x_{t+1}^{b}$.

This extra step amounts to update the ancestral lineage of the selected
trajectory up to time $t$, recursively in time, from $t=T-1$ to
time $t=0$. It is straightforward to show that \prettyref{eq:BS_distr}
is the conditional distribution of random variable $B_{t}^{\star
}=A_{t}^{B_{t+1}^{\star}}$,
conditional on $B_{t+1}^{\star}=b$\vspace*{1.5pt} and the other auxiliary variables
of the particle system, relative to the joint distribution \prettyref
{eq:defpi_xt_notations}.
As such, this extra step leaves $\mathbb{Q}_{T}(\mathrm{d}x_{0:T})$ invariant.

It should be noted that the BS step may be implemented only when the
density $m_{t+1}(x_{t},x_{t+1})$ admits an explicit expression, which
is unfortunately not the case for several models of practical interest.
Finally, Remark~\ref{rem:Equivalent}
also applies to the CPF-BS kernel.
%
\begin{remark}
\label{rem:CPFimageWithBS}Let {{$P_{T}^{N,B}$}}
denote the CPF-BS (CPF with
backward sampling) Markov kernel. Then the image of {{$z_{T}^{\star
}\in
\mathcal{X}^{T+1}$}}
under {{$P_{T}^{N,B}$}} is unchanged
by the choice of {{$(b_{0:T-1}^{\star},n^{\star})$}}
for the realization of {{$(B_{0:T-1}^{\star},N^{\star})$}}
in the initialization, that is, step CPF-1.
\end{remark}

\subsection{Reversibility, covariance ordering and asymptotic efficiency}
\label{sub:Reversibility}

To compare PG with backward sampling with PG without backward sampling,
one might be tempted to use Peskun ordering. The following counter-example
shows that unfortunately the former does not dominate the latter in
the Peskun sense.

\begin{example}
Let $\mathcal{X}=\mathbb{R}$, $(X_{t})_{t\geq0}$ be an i.i.d. sequence
with marginal law $m(\mathrm{d}x)$, and let the potentials be unit valued,
that is, $G_{t}(x_{t})=1$ for all $t$. Then one can show that the CPF
kernel with backward sampling does not dominate the CPF kernel in
Peskun sense. For example, let $\mathcal{B}=
\mathcal{B}_{0}\times\cdots\times\mathcal{B}_{T}\subset\mathcal
{X}^{T+1}$,
where $m(\mathcal{B}_{t})=\varepsilon$ for all $t$. If we choose a
reference trajectory $x_{0:T}\notin\mathcal{B}$ but $x_{t}\in
\mathcal{B}_{t}$
for $t\neq T$, then it is easy to show that (e.g., when $T=2$ and
$N=2$) that the probability of hitting $\mathcal{B}$ when starting
from $x_{0:T}$, that is, $P_{T}^{N}(x_{0:T},\mathcal{B})$, is higher
without backward sampling than with it. In this example, a chosen
trajectory that coalesces with the reference trajectory has more chance
of hitting set $\mathcal{B}$.
\end{example}

One does observe in practice that Backward sampling (BS) brings improvement
to the decay of the autocorrelation function of successive samples
of $X_{0:T}$ generated by the particle Gibbs sampler, that is, more rapid
decay compared to not implementing BS; see \cite{Lindsten2012} and
our numerical experiments in \prettyref{sec:Numerical-experiments}.
However, how much improvement depends on the transition kernel \emph
{$m_{t}(x_{t-1},\mathrm{d}x_{t})$}
of the hidden state process $ (X_{t} )_{t\geq0}$. If only
$X_{0}\sim m_{0}$ is random while \emph{$m_{t}(x_{t-1},\mathrm{d}x_{t})$}
is a point mass at $x_{t-1}$ for \emph{$t\geq1$}, then it is clear
that BS will bring no improvement. We can however prove that, regardless
of $m_{t}$, the empirical average of the successive samples from
a CPF kernel with BS will have an asymptotic variance no larger than
the asymptotic variance of the empirical average of successive samples
from the corresponding CPF kernel without BS. The asymptotic variance
here is the variance of the limiting Gaussian distribution characterised
by the usual $\sqrt{n}$-central limit theorem (CLT) for Markov chains;
see, for example, \cite{roberts2004general}.

The following result due to \cite{tierney1998note} (see also \cite{MiraGeyer1999}),
formalises this comparison, or ordering, of two Markov transition
kernels having the same stationary measure via the asymptotic variance
given by the CLT. We call this efficiency ordering.

\begin{thmm}[(\cite{tierney1998note})]
\label{thmm:tier98} For $\xi_{0},\xi
_{1},\ldots$
successive samples from a reversible Markov transition kernel $H$
(on some general state space) with stationary measure $\pi$, where
$\xi_{0}\sim\pi$, and for $f\in L^{2}(\pi)= \{ f\dvtx \int\pi
(\mathrm{d}x)f(x)^{2}<\infty \} $
let
\[
v(f,H)=\lim_{n\rightarrow\infty}\frac{1}{n}\operatorname {Variance} \Biggl(\sum
_{i=0}^{n-1}f(\xi_{i}) \Biggr).
\]
Let $P_{1}$ and $P_{2}$ be two reversible Markov kernels with stationary
measure $\pi$ such that
$\E_{\pi\otimes P_{1}} \{ g(\xi_{0})g(\xi_{1}) \} \leq
\E_{\pi
\otimes P_{2}} \{ g(\xi_{0})g(\xi_{1}) \} $
for all $g\in L^{2}(\pi)$. Then $v(f,P_{1})\leq v(f,P_{2})$ for
all $f\in L^{2}(\pi)$ and $P_{1}$ is said to dominate $P_{2}$ in
\emph{efficiency ordering}.
\end{thmm}

Note that in the original version of this theorem by \cite{tierney1998note}
the requirement on $P_{1}$ and $P_{2}$ for $v(f,P_{1})\leq v(f,P_{2})$
for all $f\in L^{2}(\pi)$ is that $P_{1}$ dominates $P_{2}$ in
Peskun ordering. However, Tierney's proof actually makes use of the
weaker Peskun implied property of lag-1 domination instead, as also
noted in Theorem~4.2 of \cite{MiraGeyer1999}.

To prove efficiency ordering, we must prove first that the CPF kernel,
with or without BS, is reversible. Following reversibility, we then
need to show that the CPF kernel with BS has smaller lag 1 autocorrelation
compared to the CPF kernel without BS; this property if holds is called
\emph{lag-one domination}.

\begin{prop}
The CPF kernel is reversible. \label{prop:reversibility_CPF}
\end{prop}

\begin{pf}
This result is based on the equivalent representation of the CPF kernel
described in Section~\ref{sec:PG_framework} (see Remark~\ref
{rem:Equivalent}) which regenerates both
the labels $(B_{0:T-1}^\star,N^\star)$ of the frozen trajectory, and
the $N-1$ remaining
trajectories $(X_{0:T}^{1:N\setminus\star},A_{0:T-1}^{1:N\setminus
\star})$. Consider a measurable function $h\dvtx \setX^{T+1}\times\setX
^{T+1}\rightarrow\mathbb{R}$
and let $(Z_{T}^{\star},\check{Z}_{T}^{\star})=(X_{0:T}^{\star
},\check
{X}_{0:T}^{\star})\sim\mathbb{Q}_{T}\otimes P_{T}^{N}$
then
\begin{eqnarray*}
\E \bigl\{ h\bigl(Z_{T}^{\star},\check{Z}_{T}^{\star}
\bigr) \bigr\} & = & \int \mathbb{Q}_{T}\bigl(\mathrm{d}x_{0:T}^{\star}
\bigr)\int\pi _{T}^{N}\bigl(\mathrm{d}x_{0:T}^{1:N\setminus
\star},\mathrm{d}a_{0:T-1}^{1:N\setminus\star},\mathrm{d}b_{0:T-1}^{\star},\mathrm{d}n^{\star
}|x_{0:T}^{\star}
\bigr)
\\
& &{} \times\int\pi_{T}^{N} \bigl(\mathrm{d}\check{n}^{\star
}|x_{0:T}^{1:N},a_{0:T-1}^{1:N}
\bigr)h\bigl(z_{T}^{n^{\star
}},z_{T}^{\check
{n}^{\star}}
\bigr)
\\
& = & \int\pi_{T}^{N}\bigl(\mathrm{d}x_{0:T}^{1:N},\mathrm{d}a_{0:T-1}^{1:N},\mathrm{d}n^{\star
}
\bigr)\int \pi_{T}^{N}\bigl(\mathrm{d}\check{n}^{\star
}|x_{0:T}^{1:N},a_{0:T-1}^{1:N}
\bigr)h\bigl(z_{T}^{n^{\star}},z_{T}^{\check
{n}^{\star}}\bigr)
\\
& = & \int\pi_{T}^{N}\bigl(\mathrm{d}x_{0:T}^{1:N},\mathrm{d}a_{0:T-1}^{1:N},\mathrm{d}
\check {n}^{\star
}\bigr)\int\pi_{T}^{N}
\bigl(\mathrm{d}n^{\star
}|x_{0:T}^{1:N},a_{0:T-1}^{1:N}
\bigr)h\bigl(z_{T}^{n^{\star}},z_{T}^{\check
{n}^{\star}}\bigr)
\\
& = & \E \bigl\{ h\bigl(\check{Z}_{T}^{\star},Z_{T}^{\star}
\bigr) \bigr\},
\end{eqnarray*}
where the second equality uses Remark~\ref{rem:Equivalent}. We have
also used a change of variables, that is, $z_{T}^{n^{\star}}$,
respectively, $z_{T}^{\check{n}^{\star}}$, must be understood as a certain
deterministic function of $(x_{0:T}^{1:N},a_{0:T-1}^{1:N},n^{\star})$,
respectively, $(x_{0:T}^{1:N},a_{0:T-1}^{1:N},\check{n}^{\star})$, in the
equations above; see notation at the beginning of this section,
specifically in the paragraph before step CPF-3.
\end{pf}
We now prove a similar result for CPF-BS, the CPF kernel with backward
sampling.

\begin{prop}
\label{prop:BSreversible}
The CPF-BS kernel (CPF with
backward sampling) is reversible.
\end{prop}

\begin{pf}
Consider a $\setX^{T+1}\times\setX^{T+1}\rightarrow\mathbb{R}$ measurable
function $h$ and let $(Z_{T}^{\star},\check{Z}_{T}^{\star
})= (X_{0:T}^{\star},\check{X}_{0:T}^{\star})\sim\mathbb
{Q}_{T}\otimes
P_{T}^{N,B}$. To evaluate
$\E \{ h(Z_{T}^{\star},\check{Z}_{T}^{\star}) \}$, we first
invoke Remark~\ref{rem:CPFimageWithBS} and then the following observation:
the image of CPF-BS would be unchanged
if step CPF-3 would be replaced by a Gibbs step that would update
the complete genealogy, that is, replace $A_{0:T-1}^{1:N}$ by $\check
{A}_{0:T-1}^{1:N}$,
a sample from $\pi_{T}^{N}(\mathrm{d}a_{0:T-1}^{1:N}|x_{0:T}^{1:N},n^{\star})$.
This is because the $A_{t}^{n}$'s are
independent conditionally on $(X_{0:T}^{1:N},N^{\star})$.
Thus,
\begin{eqnarray*}
\E \bigl\{ h\bigl(Z_{T}^{\star},\check{Z}_{T}^{\star}
\bigr) \bigr\} 
&
=&\int\pi_{T}^{N}\bigl(\mathrm{d}x_{0:T}^{1:N},\mathrm{d}a_{0:T-1}^{1:N},\mathrm{d}n^{\star}
\bigr)
\\
&&{} \times\int\pi_{T}^{N}\bigl(\mathrm{d}\check{n}^{\star}|x_{0:T}^{1:N}
\bigr)\int\pi _{T}^{N}\bigl(\mathrm{d}\check{a}_{0:T-1}^{1:N}|x_{0:T}^{1:N},
\check{n}^{\star
}\bigr)h\bigl(z_{T}^{n^{\star}},z_{T}^{\check{n}^{\star}}
\bigr)
\\
& =&\int\pi_{T}^{N}\bigl(\mathrm{d}x_{0:T}^{1:N},\mathrm{d}
\overline{a}_{0:T-1}^{1:N}\bigr)\int \pi _{T}^{N}
\bigl(\mathrm{d}n^{\star}|x_{0:T}^{1:N}\bigr)\pi
_{T}^{N}\bigl(\mathrm{d}a_{0:T-1}^{1:N}|x_{0:T}^{1:N},n^{\star}
\bigr)
\\
&&{} \times\int\pi_{T}^{N}\bigl(\mathrm{d}\check{n}^{\star}|x_{0:T}^{1:N}
\bigr)\int\pi _{T}^{N}\bigl(\mathrm{d}\check{a}_{0:T-1}^{1:N}|x_{0:T}^{1:N},
\check{n}^{\star
}\bigr)h\bigl(z_{T}^{n^{\star}},z_{T}^{\check{n}^{\star}}
\bigr)
\\
& =&\E \bigl\{ h\bigl(\check{Z}_{T}^{\star},Z_{T}^{\star}
\bigr) \bigr\},
\end{eqnarray*}
where the second equality is based on the fact that one may generate
$(X_{0:T}^{1:N},A_{0:T-1}^{1:N},N^{\star})\sim\pi_{T}^{N}$
as: $(X_{0:T}^{1:N},\overline{A}_{0:T-1}^{1:N},N^{\star})\sim\pi_{T}^{N}$,
then update $\overline{A}_{0:T-1}^{1:N}$ as $A_{0:T-1}^{1:N}$ through
the Gibbs step $\pi_{T}^{N}(\mathrm{d}a_{0:T-1}^{1:N}|x_{0:T}^{1:N},n^{\star})$,
and the third equality is a simple change of variables. The simplification
of $\pi_{T}^{N}(\mathrm{d}\check{n}^{\star}|x_{0:T}^{1:N},a_{0:T-1}^{1:N})$
into $\pi_{T}^{N}(\mathrm{d}\check{n}^{\star}|x_{0:T}^{1:N})$ (first equality
onward) reflects the fact that step~CPF-2 does not depend on $a_{0:T-1}^{1:N}$.
\end{pf}
The final result shows that the CPF-BS kernel dominates the CPF kernel
in lag-one autocorrelation.

\begin{thmm}
\label{prop:multiBSlag1domination} The CPF-BS kernel,
denoted $P_{T}^{N,B}$, dominates the CPF kernel in lag one autocorrelation,
that is, let $h$ be square integrable function \textup{then}
\[
0\leq\E_{\mathbb{Q}_{T}\otimes P_{T}^{N,B}} \bigl\{ h\bigl(Z_{T}^{\star
}\bigr)h
\bigl(\check{Z}_{T}^{\star}\bigr) \bigr\} \leq\E_{\mathbb{Q}_{T}\otimes
P_{T}^{N}}
\bigl\{ h\bigl(Z_{T}^{\star}\bigr)h\bigl(\check{Z}_{T}^{\star}
\bigr) \bigr\}.
\]
\end{thmm}

\begin{pf}
We use again the facts that $\pi_{T}^{N}(\mathrm{d}\check{n}^{\star
}|x_{0:T}^{1:N},a_{0:T-1}^{1:N})$
reduces into $\pi_{T}^{N}(\mathrm{d}\check{n}^{\star}|x_{0:T}^{1:N})$, and
that, under multinomial resampling, step CPF-3 may be replaced by
a Gibbs step that updates the complete genealogy as in the proof of
Proposition~\ref{prop:BSreversible}.
\begin{eqnarray*}
&&\E_{\mathbb{Q}_{T}\otimes P_{T}^{N,B}} \bigl\{ h\bigl(Z_{T}^{\star
}\bigr)h\bigl(
\check {Z}_{T}^{\star}\bigr) \bigr\} \\
&&\quad =\int
\pi_{T}^{N}\bigl(\mathrm{d}x_{0:T}^{1:N}\bigr)\int
\pi _{T}^{N}\bigl(\mathrm{d}n^{\star}|x_{0:T}^{1:N}
\bigr)\int\pi _{T}^{N}\bigl(\mathrm{d}a_{0:T-1}^{1:N}|x_{0:T}^{1:N},n^{\star}
\bigr)
\\
&&\qquad{} \times\int\pi_{T}^{N}\bigl(\mathrm{d}\check{n}^{\star}|x_{0:T}^{1:N}
\bigr)\int\pi _{T}^{N}\bigl(\mathrm{d}\check{a}_{0:T-1}^{1:N}|x_{0:T}^{1:N},
\check{n}^{\star
}\bigr)h\bigl(z_{T}^{n^{\star}}\bigr)h
\bigl(z_{T}^{\check{n}^{\star}}\bigr)
\\
&&\quad =\int\pi_{T}^{N}\bigl(\mathrm{d}x_{0:T}^{1:N}
\bigr) \biggl(\int\pi_{T}^{N}\bigl(\mathrm{d}n^{\star
}|x_{0:T}^{1:N}
\bigr)\int\pi _{T}^{N}\bigl(\mathrm{d}a_{0:T-1}^{1:N}|x_{0:T}^{1:N},n^{\star
}
\bigr)h\bigl(z_{T}^{n^{\star}}\bigr) \biggr)^{2}
\\
&&\quad \leq\int\pi_{T}^{N}\bigl(\mathrm{d}x_{0:T}^{1:N}
\bigr)\int\pi _{T}^{N}\bigl(\mathrm{d}a_{0:T-1}^{1:N}|\,\mathrm{d}x_{0:T}^{1:N}
\bigr) \biggl(\int\pi _{T}^{N}\bigl(\mathrm{d}n^{\star
}|x_{0:T}^{1:N},a_{0:T-1}^{1:N}
\bigr)h\bigl(z_{T}^{n^{\star}}\bigr) \biggr)^{2}
\\
&&\quad =\E_{\mathbb{Q}_{T}\otimes P_{T}^{N}} \bigl\{ h\bigl(Z_{T}^{\star
}\bigr)h
\bigl(\check {Z}_{T}^{\star}\bigr) \bigr\}.
\end{eqnarray*}
The penultimate line uses Jensen inequality. The last line is indeed
the same expectation but under $\mathbb{Q}_{T}\otimes P_{T}^{N}$
(no BS step).
\end{pf}

We are now in position to state the main result of this section.

\begin{thmm}
The CPF-BS kernel dominates the CPF kernel in efficiency ordering.
\end{thmm}

\begin{pf}
This is a direct consequence of Theorem \ref{thmm:tier98}
and Propositions \ref{prop:reversibility_CPF} and \ref{prop:BSreversible}.
\end{pf}

\section{Alternative resampling schemes}
\label{sec:alt-resampling}

We mentioned in the previous section that the backward sampling step is
not always
applicable, as it relies on the probability density of the Markov
kernel $m_t$ being tractable.
In this section, we discuss another way to improve the performance of
particle Gibbs through the introduction of resampling schemes that are
less noisy than multinomial resampling. The intuition is
that such resampling schemes tend
to reduce path degeneracy in particle systems, and thus should lead to
better mixing for particle Gibbs; see, for example, \cite{pathstorage}
for some results on path degeneracy.

We no longer assume that the resampling
distribution $\varrho_t$ is \eqref{eq:multiRS}, and we rewrite $\pi
_T^N$ under
the more general expression (using the same notation as in Section~\ref{sec:PG_framework})
%
\begin{eqnarray}\label
{eq:defpi_general}
&&\pi_{T}^{N} \bigl(\mathrm{d}z_{0:T}^{1:N},\mathrm{d}a_{0:T-1}^{1:N},\mathrm{d}n^\star
\bigr) \nonumber\\
&&\quad =\frac{1}{\mathcal{Z}_T}m_{0}^{\otimes N}\bigl(\mathrm{d}z_{0}^{1:N}
\bigr)
\\
&&\qquad{} \times\prod_{t=1}^{T}
\Biggl\{ \Biggl[\frac{1}{N}\sum_{n=1}^{N}G_{t-1}
\bigl(z_{t-1}^{n}\bigr) \Biggr] \R_{t-1}
\bigl(z_{t-1}^{1:N},\mathrm{d}a_{t-1}^{1:N}\bigr)
\prod_{n=1}^{N} \bigl[ q_{t}
\bigl(z_{t-1}^{a_{t-1}^{n}},\mathrm{d}z_{t}^{n}\bigr)
\bigr] \Biggr\}\frac{1}{N}G_{T}\bigl(z_{T}^{n^\star}
\bigr).\nonumber
\end{eqnarray}

Recall that to establish validity of particle Gibbs, we applied the
following change of variables:
\[
\bigl(z_{0:T}^{1:N},a_{0:T-1}^{1:N},n^{\star}
\bigr) \leftrightarrow\bigl(z_{0:T}^{1:N\setminus\star
},a_{0:T-1}^{1:N\setminus
\star},z_{0:T}^{\star},
b_{0:T-1}^{\star},n^{\star}\bigr),
\]
to $\pi_T^N$, which led to distribution \eqref{eq:pi_TN_reformulated},
which is such that $Z_T^\star=Z_T^{N^\star}$ has marginal distribution
$\mathbb{Q}_T(\mathrm{d}z_T)$. To generalise \eqref{eq:pi_TN_reformulated}
to resampling schemes other than multinomial resampling, we assume that
the resampling
distribution is \emph{marginally unbiased}: the joint distribution
$\varrho_t(z_t^{1:N},\mathrm{d}a_t^{1:N})$ (for fixed $z_t^{1:N}$) is such that the
marginal distribution
of a single component $A_t^n$ is the discrete distribution which
assigns probability $W_t^m$ to outcome $m\in1\dvtx N$. (We shall see that,
up to a trivial modification, standard resampling schemes fulfil this
condition.)

Under marginal unbiasedness, $\varrho_t(z_t^{1:N},\mathrm{d}a_t^{1:N})$ may be
decomposed as follows, for any
$n\in1\dvtx N$:
\[
\varrho_t\bigl(z_t^{1:N},\mathrm{d}a_t^{1:N}
\bigr) = \bigl\{W_t^{a_t^n}\bigl(z_t^{1:N}
\bigr) \,\mathrm{d}a_t^n \bigr\} \varrho_t^c
\bigl(z_t^{1:N}, \bigl\{\mathrm{d}a_t^{1:N\setminus
n}|A_t^n=a_t^n
\bigr\} \bigr), %
\]
where the second factor above denotes the distribution of the $(N-1)$
labels $A_t^{1:N\setminus n}$ conditional on $A_t^{n}=a_t^n$ which
corresponds to the joint distribution $\varrho_t(z_t^{1:N},\mathrm{d}a_t^{1:N})$.
Thus, applying the change of variables above to $\pi_T^N$ gives
%
\begin{eqnarray}
&&\pi_{T}^{N}\bigl(\mathrm{d}z_{0:T}^{1:N\setminus\star},\mathrm{d}a_{0:T-1}^{1:N\setminus
\star
},\mathrm{d}z_{0:T}^{\star},\mathrm{d}b_{0:T-1}^{\star},\mathrm{d}n^{\star}
\bigr)\nonumber\\
&&\quad=
\frac{1}{N^{T+1}}\bigl(\mathrm{d}b_{0:T-1}^{\star}\,\mathrm{d}n^{\star}\bigr)
\mathbb {Q}_{T}\bigl(\mathrm{d}z_{T}^{\star}\bigr)\prod
_{t=0}^{T-1}\delta_{ (
[z_{T}^{\star
} ]_{t+1} )}
\bigl(\mathrm{d}z_{t}^{\star}\bigr)
\\
&&\qquad{}\times\prod_{n\neq b_{0}^{\star}}m_{0}
\bigl(\mathrm{d}z_{0}^{n}\bigr) \Biggl[\prod
_{t=1}^{T} \R_{t-1}^c
\bigl(z_{t-1}^{1:N}, \bigl\{\mathrm{d}a_{t-1}^{1:N \setminus
b_t^\star
} |
A_{t-1}^{b_t^\star}= b_{t-1}^{\star} \bigr\} \bigr)
\prod_{n\neq b_{t}^{\star}} q_{t}\bigl(z_{t-1}^{a_{t-1}^{n}},\mathrm{d}z_{t}^{n}
\bigr) \Biggr].
\nonumber
\end{eqnarray}

Inspection of the distribution above reveals that step CPF-1 (as
defined in Section~\ref{sec:PG_framework}), that is, the Gibbs step
that regenerates the complete particle system
conditional on one ``frozen'' trajectory~$z_{0:T}^\star$, now requires
to sample
at each iteration $t$ from the conditional resampling distribution
$\varrho^c
_{t-1}$. The two next sections
explains how to do so for two popular resampling schemes,
namely residual resampling and systematic resampling.

To simplify notation in the next sections, we will remove any
dependency in $t$, and consider
the generic problem of deriving, from a certain distribution of $N$
labels $A^{1:N}$ based
on normalised weights $W^{1:N}$, the conditional distribution of
$A^{1:N}$ given that one component is fixed.

\subsection{Conditional residual resampling}
\label{sub:Conditional-residual-resampling}


The standard definition of residual resampling \cite{LiuChen} is
recalled as
Algorithm \ref{alg:res}. It is clear that this resampling scheme is
such that the number
of off-springs of particles $n$ is a random variable with expectation
$NW^n$ (assuming $W^{1:N}$ is the vector of the $N$ normalised weights
used as input).
To obtain a resampling distribution that is marginally unbiased (as
defined in the previous
section), we propose the following simple modification: we run
Algorithm \ref{alg:res},
and then we permute randomly the output: $A^{1:N}=\bar{A}^{\sigma
(1:N)}$ where $\sigma$ is chosen uniformly among the $N!$ permutations
on the set $1\dvtx N$.

\begin{algorithm}[b]
\caption{Residual resampling}
\label{alg:res}
\begin{description}
\item[\textbf{Input:}] normalised weights $W^{1:N}$
\item[\textbf{Output:}] a vector of $N$ random labels $\bar
{A}^{1:N}\in1\dvtx N$
\begin{enumerate}[(a)]
\item[\textbf{(a)}] Compute $r^n = NW^n - \lfloor NW^n \rfloor$ (for each
$n\in
1\dvtx N)$ and $R=\sum_{n=1}^N r^n$.
\item[\textbf{(b)}] Construct $\bar{A}^{1:(N-R)}$ as the ordered vector of size
$(N-R)$ that contains $\lfloor NW^n \rfloor$ copies of value $n$ for
each $n\in1\dvtx N$.
\item[\textbf{(c)}] For each $n\in(N-R+1)\dvtx N$, sample $\bar{A}^n\sim\mathcal
{M}(r^{1:N}/R)$.
\end{enumerate}
\end{description}
\end{algorithm}

Another advantage of randomly permuting the labels obtained by residual
resampling is that
it makes the particle system exchangeable, as with multinomial
resampling (assuming that
residual resampling is applied at every iteration $t$ of the particle
algorithm). Thus, using
the same line of reasoning as in Section~\ref{sec:PG_framework}, we see
that one may arbitrarily
relabel the frozen trajectory $z_T^\star$ as $(1,\ldots,1)$ before
applying step CPF-1.
Therefore, it is sufficient
to derive an algorithm to sample from the distribution of labels
$A^{2:N}$, conditional
on $A^1=1$.

We observe that, under residual resampling (with randomly permuted
output), the probability
that $A^1$ is set to one of the $ \lfloor NW^{1} \rfloor$
$ $ ``deterministic'' copies of label $1$ is $ \lfloor
NW^{1}
\rfloor/NW^{1}$.
This remark leads to Algorithm \ref{alg:rescond}, which generates
a vector $A^{1:N}$ of $N$ labels such that $A^1=1$.

In practice, assuming conditional residual resampling is applied at
every iteration of the particle algorithm (i.e., when generating
$(X_{0:T}^{2:N},A_{0:T-1}^{2:N})$ conditional on $X_{0:T}^1$), step (d)
of Algorithm~\ref{alg:rescond} may be
omitted, as the actual order of particles with labels $2\dvtx N$ do not play
any role in the following iterations (and, therefore, has no bearing on
the image of the CPF kernel).

\begin{algorithm}[t]
\caption{Conditional residual resampling}
\label{alg:rescond}
\begin{description}
\item[\textbf{Input:}] normalised weights $W^{1:N}$
\item[\textbf{Output:}] a vector of $N$ random labels $A^{1:N}\in1\dvtx N$
such that $A^1=1$
\begin{enumerate}[(a)]
\item[\textbf{(a)}] Compute $r^n = NW^n - \lfloor NW^n \rfloor$ (for each
$n\in
1\dvtx N)$ and $R=\sum_{n=1}^N r^n$.
\item[\textbf{(b)}] Generate $U\sim\mathcal{U}[0,1]$.
\item[\textbf{(c)}] \textbf{If} $U< \lfloor NW^{1} \rfloor/NW^{1}$,
\textbf{then} generate $\bar{A}^{1:N}$ using
Algorithm \ref{alg:res};
\item[] \textbf{Else} generate $\bar{A}^{2:N}$ exactly
as in Algorithm \ref{alg:res}, except
that the number of multinomial draws in step (c) is $R-1$ instead of
$R$. (Thus $\bar{A}^{2:N}$
contains $\lfloor NW^n \rfloor$ deterministic copies of value $n$, for
each $n\in1\dvtx N$, and
$(R-1)$ random copies.)
\item[\textbf{(d)}] Let $A^1=1$, and $A^{2:N}=\bar{A}^{\sigma(2:N)}$, where
$\sigma$ is a random $2\dvtx N\rightarrow2\dvtx N$ permutation.
\end{enumerate}
\end{description}
\end{algorithm}

\subsection{Conditional systematic resampling}
\label{sub:Conditional-systematic-resampling}

The systematic resampling algorithm of \cite{CarClifFearn} consists
in creating $N$ off-springs, based on the normalised weights $W^{1:N}$, as
follows. Let $U$ a uniform variate in $[0,1]$, $v^0=0$, $v^n=\sum_{m=1}^n NW^m $, and set $\bar{A}^n=m$ for the $N$ pairs $(n,m)$ such
that $v^{m-1}\leq U+n-1 < v^m$.
The standard algorithm to perform systematic resampling (for a given
$U$, sampled from $\mathcal{U}([0,1])$ beforehand) is recalled as
Algorithm \ref{alg:sys}.

\begin{algorithm}
\caption{Systematic resampling (for a given $U$)}
\label{alg:sys}
\begin{description}
\item[\textbf{Input:}] normalised weights $W^{1:N}$, and $U\in[0,1]$
\item[\textbf{Output:}] a vector of $N$ random labels $A^{1:N}\in1\dvtx N$
\begin{enumerate}[(a)]
\item[\textbf{(a)}] Compute the cumulative weights as 
$v^n=\sum_{m=1}^n NW^m$ for $n\in1\dvtx N$.
\item[\textbf{(b)}] Set $s\leftarrow U$, $m\leftarrow1$.
\item[\textbf{(c)}] \textbf{For} $n= 1\dvtx N$
\item[] \quad\textbf{While} $v^m<s$ \textbf{do} $m\leftarrow m+1$.
\item[]\quad$\bar{A}^n\leftarrow m$, and $s\leftarrow s+1$.
\item[]$ $ \textbf{End For}
\end{enumerate}
\end{description}
\end{algorithm}

To obtain a resampling distribution that is marginally unbiased, we
propose to randomly cycle the output: $A^{1:N}=\bar{A}^{c(1:N)}$, where
$c\dvtx (1\dvtx N)\rightarrow(1\dvtx N)$ is drawn uniformly among the $N$ cycles of
length $N$.
Recall that a cycle $c$ is a permutation such that for a certain $c_0
\in1\dvtx N$ and for all $n \in1\dvtx N$,
$c(n)=c_0+n$ if $c_0+n\leq N$, $c(n)=c_0+n-N$ otherwise.

Cycle randomisation is slightly more convenient that permutation
randomisation when it comes to deriving the conditional systematic
resampling algorithm. It is also slightly cheaper in computational
terms. Under cycle randomisation, the particle system is no longer
exchangeable (assuming systematic resampling is carried out at each
iteration), but it remains true that one has the liberty to relabel
arbitrarily the
frozen trajectory, say with labels $(1,\ldots,1)$, without changing the
image of the PG kernel. (A proof may be obtained from the corresponding
author.) Thus, as in the previous section,
it is sufficient to derive the algorithm to simulate $A^{2:N}$
conditional on $A^1=1$.

A distinctive property of systematic resampling is that
the number of off-springs of particle $n$ is either $ \lfloor
NW^{n} \rfloor$
or $ \lfloor NW^{n} \rfloor+1$. In particular, the
algorithm starts
by creating $\lfloor N W^1 \rfloor$ ``deterministic'' copies of
particle $1$, then
adds one extra ``random copy,'' with probability $r^1=NW^1-\lfloor
NW^1\rfloor$, and so on.
When $N$ off-springs have been obtained, the output is randomly cycled.
Thus, conditional on $A^1=1$, the probability that a deterministic
copy of $1$ was moved to position $1$ is $ \lfloor NW^{1}
\rfloor/NW^{1}$.
This observation leads to the Algorithm \ref{alg:syscond} for generating from
$A^{2:N}$ conditional on $A^1=1$.

\begin{algorithm}[b]
\caption{Conditional systematic resampling}
\label{alg:syscond}
\begin{description}
\item[\textbf{Input:}] normalised weights $W^{1:N}$
\item[\textbf{Output:}] a vector of $N$ random labels $A^{1:N}\in1\dvtx N$
such that $A^1=1$
\begin{enumerate}[(a)]
\item[\textbf{(a)}] \textbf{If} $NW^1\leq1$, sample $U\sim\mathcal
{U}[0,NW^1] $.
\item[$ $] \textbf{Else} Set $r^1=NW^1-\lfloor NW^1 \rfloor$. With
probability $\frac{r^{1} ( \lfloor NW^{1} \rfloor
+1
)}{NW^{1}}$, sample $U\sim\mathcal{U}[0,r^{1}]$, otherwise sample
$U\sim
\mathcal{U}[r^{1},1]$.
\item[\textbf{(b)}] Run Algorithm \ref{alg:sys} with inputs $W^{1:N}$ and $U$;
call $\bar{A}^{1:N}$
the output.
\item[\textbf{(c)}] Choose $C$ uniformly from the set of cycles such that $\bar
{A}^{C(1)}=1$
and set $A^{1:N}=\bar{A}^{C(1:N)}$.
\end{enumerate}
\end{description}
\end{algorithm}

\subsection{Note on backward sampling for alternative resampling schemes}

It is possible to adapt the backward sampling step (see Section~\ref{sec:Backward-sampling}) to residual or systematic resampling.
Unfortunately, the corresponding algorithmic details are quite involved,
and the results are not very satisfactory (in the sense of not
improving strongly the mixing of particle Gibbs relative to the version
without a backward sampling step); see \cite{PGibbs_arxiv} for details.
This seems related to the strong dependence between the labels
$A_t^{1:N}$ that is induced by residual resampling and particularly
systematic resampling, which therefore makes it more difficult to update
one single component of this vector.

As pointed out by a referee, it is straightforward to adapt
\emph{ancestor sampling} \cite{LindstenAncestor}, which is an
alternative approach to backward sampling for rejuvenating the ancestry
of the frozen trajectory, to the alternative resampling schemes.
The mixing gains of doing so deserves further investigation.

\section{Numerical experiments}
\label{sec:Numerical-experiments}

The focus of our numerical experiments is on comparing the four variants
of particle Gibbs discussed in this paper, corresponding to the three
different resampling schemes (multinomial, residual, systematic),
and whether the extra backward sampling is performed or not (assuming
multinomial resampling).

We consider the following state-space model:
\[
X_{0}\sim N\bigl(\mu,\sigma^{2}\bigr),\qquad
X_{t+1}|X_{t}=x_{t}\sim N \bigl(\mu +
\rho(x_{t}-\mu),\sigma^{2} \bigr),\qquad Y_{t}|X_{t}=x_{t}
\sim\operatorname{Poisson}\bigl(\mathrm{e}^{x_{t}}\bigr)
\]
for $t\in0\dvtx T$, hence one may take $G_{t}(x_{t})=\exp \{
-\mathrm{e}^{x_{t}}+y_{t}x_{t} \} $,
where $y_{t}$ is the observed value of $Y_{t}$. This model is motivated
by \cite{yu2011center} who consider a similar model for photon counts
in X-ray astrophysics. The parameters $\mu$, $\rho$, $\sigma$ are
assumed to be unknown, and are assigned the following (independent)
prior distributions: $\rho\sim\operatorname{Uniform}[-1,1]$, $\mu\sim
N(m_{\mu
},s_{\mu}^{2})$
and $1/\sigma^{2}\sim\operatorname{Gamma}(a_{\sigma},b_{\sigma})$; let
$\theta=(\mu,\rho,\sigma^{2})$. (We took $m_{\mu}=0$, $s_{\mu}=10$,
$a_{\sigma}=b_{\sigma}=1$ in our simulations.) We run a Gibbs sampler
that targets the posterior distribution of $(\theta,X_{0:T})$, conditional
on $Y_{0:T}=y_{0:T}$, by iterating (a) the Gibbs step that samples
from $\theta|X_{0:T},Y_{0:T}$, described below; and (b) the particle
Gibbs step discussed in this paper, which samples from $X_{0:T}|\theta
,Y_{0:T}$.
Direct calculations show that step (a) may be decomposed into the
following successive three operations, which sample from the full
conditional distribution of each component of $\theta$, conditional
on the other components of $\theta$ and $X_{0:T}$:
\begin{eqnarray*}
1/\sigma^{2}|X_{0:T}&=&x_{0:T},Y_{0:T},
\mu,\rho\sim\operatorname{Gamma} \Biggl(a_{\sigma}+\frac{T+1}{2},b_{\sigma}+
\frac{1}{2}\tilde {x}_{0}^{2}+\frac
{1}{2}\sum
_{t=0}^{T-1} (\tilde{x}_{t+1}-\rho
\tilde {x}_{t} )^{2} \Biggr),
\\
\rho|X_{0:T}&=&x_{0:T},Y_{0:T},\mu,\sigma\sim
N_{[-1,1]} \biggl(\frac
{\sum_{t=0}^{T-1}\tilde{x}_{t}\tilde{x}_{t+1}}{\sum_{t=0}^{T-1}\tilde
{x}_{t}^{2}},\frac{\sigma^{2}}{\sum_{t=0}^{T-1}\tilde
{x}_{t}^{2}} \biggr),
\\
\mu|X_{0:T}&=&x_{0:T},Y_{0:T},\rho,\sigma\sim N
\biggl(\frac
{1}{\lambda_{\mu
}} \biggl\{ \frac{m_{\mu}}{s_{\mu}^{2}}+\frac{x_{0}+(1-\rho)\sum_{t=0}^{T-1}(x_{t+1}-\rho x_{t})}{\sigma^{2}} \biggr\},
\frac
{1}{\lambda
_{\mu}} \biggr),
\end{eqnarray*}
where we have used the short-hand notation $\tilde{x}_{t}=x_{t}-\mu$,
\[
\lambda_{\mu}=\frac{1}{s_{\mu}^{2}}+\frac{1+T(1-\rho)^{2}}{\sigma^{2}},
\]
and where $N_{[-1,1]}(m,s^{2})$ denotes the Gaussian distribution
truncated to the interval $[-1,1]$.

In each case, we run our Gibbs sampler for $10^{5}$ iterations, and
discard the first $10^{4}$ iterations as a burn-in period. Apart
from the resampling scheme, and whether or not backward sampling is
used, the only tuning parameter for the algorithm is the number of
particles $N$ in the particle Gibbs step.

\subsection{First dataset}

The first dataset we consider is simulated from the model, with $T+1=400$,
$\mu=0$, $\rho=0.9$, $\sigma=0.5$.

Figure~\ref{fig:ACF-data1-N200} reports the ACF (Autocorrelation function)
of certain components of $ (\theta,X_{0:T} )$ for the four
considered variants of our algorithm, for $N=200$.

Clearly, the version which includes a backward sampling step performs
best. The version based on systematic resampling comes second. This
suggests that, in situations where backward sampling cannot be implemented,
one may expect that using systematic resampling should be beneficial.

\begin{figure}

\includegraphics{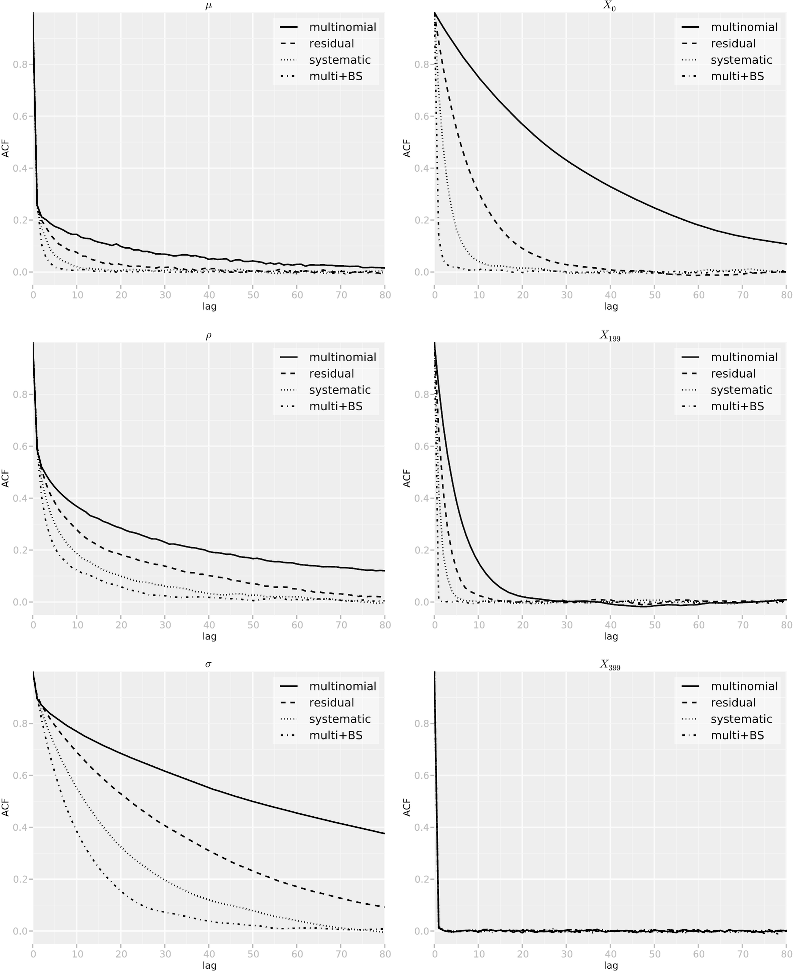}

\caption{First dataset: ACF for different
components of $(\theta,X_{0:T})$ and the four considered variants
of particle Gibbs ($N=200$).}\label{fig:ACF-data1-N200}
\end{figure}

%


It is also worthwhile to look at the update rates of $X_{t}$ with
respect to $t$ which is defined as the proportion of iterations where
$X_{t}$ changes value; see left panel of Figure~\ref{fig:coales-data1}.
This figure reveals that backward sampling increases very significantly
the probability of updating $X_{t}$ to a new value, especially at
small $t$ values, to a point where this proportion is close to one.
This also suggests that good performance for backward sampling might
be obtained with a smaller value of $N$.

To test this idea, we ran the four variants of our Gibbs sampler, but
with $N=20$. The right side of Figure~\ref{fig:coales-data1} reveals
the three non-backward sampling algorithms provide useless results
because components of $X_{0:300}$ hardly ever change values. For
the same reasons, the ACFs of these variants do not decay at reasonable
rate (which are not shown here).

To summarise, in this particular exercise, implementing backward sampling
is very beneficial, as it leads to good mixing even if $N$ is small.
If backward sampling could not be implemented, then using systematic resampling
may also improve performance, but not to same extent as backward sampling,
as it may still require to take $N$ to a larger value to obtain reasonable
performance.

\begin{figure}

\includegraphics{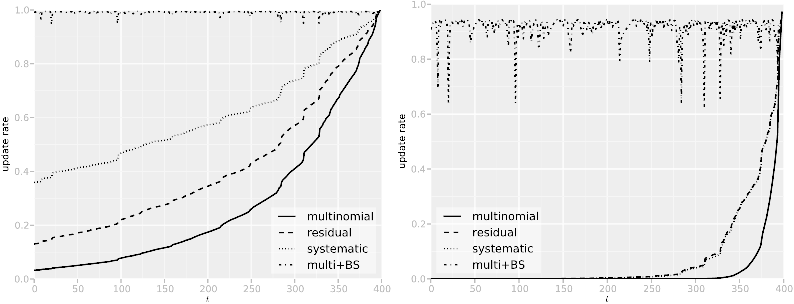}

\caption{First dataset and resulting update rates
of $X_{t}$ versus $t\in0:399$. Left plot is for $N=200$ and right
is for $N=20$.
For $N=20$, forward only versions of systematic and residual
perform similarly.} \label{fig:coales-data1}
\end{figure}

\subsection{Second dataset}

We consider a second dataset, simulated from the model with $T+1=200$,
$\mu=\log(5000)$, $\rho=0.5$, $\sigma=0.1$. (These values are close
to the posterior expectation for the real dataset of~\cite{yu2011center}.)

The interest of this example relative to the first one is twofold.
Firstly, the positive impact of backward sampling is even bigger in
this case. We have
to increase $N$ to $N=1000$ to obtain non-zero update rates for
the three variants that do not use backward sampling, whereas good
update rates may be obtained for $N=20$ for either multinomial or
residual resampling, when backward sampling is used; see Figure~\ref{fig:coal-data2}.

Secondly, we observe that backward sampling leads to excellent performance
even when $N$ is small, see also the ACF in Figure~\ref{fig:ACF-N20-two-datasets}, which
are close to the ACF of an independent process. Thus, the performance
of that variant of particle Gibbs seems to on par with the algorithm of~\cite{yu2011center}, which is specialised to this particular model
(whereas particle Gibbs may be used in a more general class of models).

\begin{figure}[b]

\includegraphics{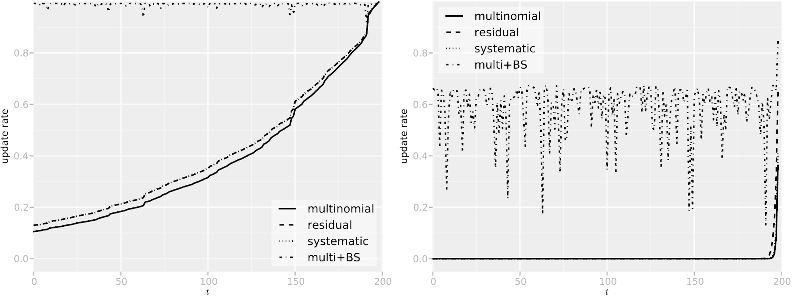}

\caption{Second dataset, same plots as Figure \protect\ref{fig:coales-data1},
with $N=1000$ (left panel) and $N=20$ (right panel). Same legend
as Figure \protect\ref{fig:ACF-data1-N200}.
In the left plot, residual and systematic are largely indistinguishable.
In the right plot, the three forward only schemes indistinguishable
before $t\approx190$.}\label{fig:coal-data2}
\end{figure}

\begin{figure}

\includegraphics{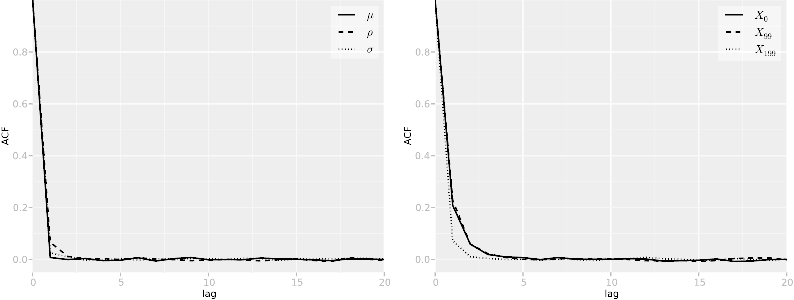}

\caption{Second dataset: ACF for certain components of $(\theta,X_{0:T})$
for particle Gibbs with backward sampling
($N=20$).}\label{fig:ACF-N20-two-datasets}
\end{figure}

\section{Discussion and conclusions}
\label{sec:Conclusion}

We now discuss the main practical conclusions that one can draw
from our numerical studies.

First, they are many situations where backward resampling cannot be
implemented, in particular when the probability density of the Markov
transition is not tractable. In that case, our simulations suggests
that one should run particle
Gibbs with systematic resampling, as this leads to better mixing.
A possible explanation is that, when only a forward pass is performed,
the lower variability of systematic resampling makes it less likely
that the proposed trajectories in the particle system coalesce with
the fixed trajectory during the resampling steps. Therefore, the particle
Gibbs step is more likely to output a trajectory which is different
than the previous one.

Second, when backward sampling can be implemented, it should be used,
as this makes it possible to set $N$ to a significantly smaller value
while maintaining good mixing; see also \cite{Lindsten2012,LindstenAncestor}
for similar findings.

In all cases, we recommend inspecting (on top of ACF plots) the same
type of plots as in Figures~\ref{fig:coales-data1} and \ref{fig:coal-data2},
that is, update rate of $X_{t}$ versus $t$, in order to assess the
mixing of the algorithm, and in particular to choose a value of $N$
that is a good trade-off between mixing properties and CPU cost. An
interesting and important theoretical line of research would be to
explain why this update rate seems more or less constant when backward
sampling is used, while it deteriorates (while going backward in time)
when backward sampling is not implemented. Another line for further
research would be to study the effect of replacing the backward sampling
step by a \emph{forward-only} ancestor sampling step as recently
proposed by \cite{LindstenAncestor}.

\section*{Acknowledgements}

We thank the editor and the referees for their insightful comments
and helpful suggestions regarding the presentation of the paper.
S.S. Singh's research was partly funded by the Engineering
and Physical Sciences Research Council (EP/G037590/1) whose support
is gratefully acknowledged.


%





\printhistory
\end{document}